\documentclass[journal,draftcls,onecolumn,10pt]{IEEEtran}

\makeatletter
\let\IEEEproof\proof
\let\IEEEendproof\endproof
\let\proof\@undefined
\let\endproof\@undefined
\makeatother

\usepackage{report}

\let\proof\IEEEproof
\let\endproof\IEEEendproof

\newcommand{\asexp}[1]{\gamma^{\alpha} \bigl[ 1 - i \beta \sgn(#1) \Phi(#1) \bigr] |#1|^{\alpha}}

\usepackage{cite}

\begin{document}

\title{\Large\bfseries On Properties of the Support of
  Capacity-Achieving Distributions for Additive Noise Channel Models
  with Input Cost Constraints\thanks{This work was supported by AUB's
    University Research Board, and the Lebanese National Council for
    Scientific Research (CNRS-L)}}
 
\author{
  \authorblockN{Jihad Fahs, Ibrahim Abou-Faycal} \\
  \authorblockA{Dept.\ of Elec.\ and Comp.\ Eng.,
    American University of Beirut \\
    Beirut 1107 2020, Lebanon \\
    {\tt \{jjf03, Ibrahim.Abou-Faycal\}@aub.edu.lb}}
}

\maketitle


\begin{abstract}

  We study the classical problem of characterizing the channel
  capacity and its achieving distribution in a generic fashion.  We
  derive a simple relation between three parameters: the input-output
  function, the input cost function and the noise probability density
  function, one which dictates the type of the optimal input. In
  Layman terms we prove that the support of the optimal input is
  bounded whenever the cost grows faster than a ``cut-off'' rate equal
  to the logarithm of the noise PDF evaluated at the input-output
  function. Furthermore, we prove a converse statement that says
  whenever the cost grows slower than the ``cut-off'' rate, the optimal
  input has necessarily an unbounded support.  In addition, we show
  how the discreteness of the optimal input is guaranteed whenever the
  triplet satisfy some analyticity properties. We argue that a
  suitable cost function to be imposed on the channel input is one
  that grows similarly to the ``cut-off'' rate.

  Our results are valid for any cost function that is
  super-logarithmic. They summarize a large number of previous channel
  capacity results and give new ones for a wide range of communication
  channel models, such as Gaussian mixtures, generalized-Gaussians and
  heavy-tailed noise models, that we state along with numerical
  computations.
  \begin{flushleft}
  \end{flushleft}
\end{abstract}


\section{Introduction}
\label{sc:Introduction}

In communication systems design, a key engineering objective is to
build systems that operate close to channel capacity.  Needless to say
that this quantity, as defined by Shannon~\cite{Sha48_1,Sha48_2} in
his pioneering work, is the cutoff value which delimits the achievable
region for ``reliable'' communications. Clearly, the channel capacity
and how it can be achieved are intimately related to the channel
model. Despite the well-known capacity results for discrete memoryless
channels, closed-form capacity expressions are rarely found in the
literature for continuous ones. The most well-understood --and perhaps
important-- continuous channel is the linear Additive White Gaussian
Channel (AWGN) subjected to an average power constraint. This AWGN
model was studied by Shannon and may be seen as an instance of the
generic real, deterministic and memoryless discrete-time additive
noise model of the form:
\begin{equation*}
  Y = f(X) + N,
\end{equation*}
where $Y \in \Reals$ is the channel output, and the channel input $X
\in \set{X} \in \Reals$ satisfies an average cost constraint of the
form $\E{\mathcal{C}\left(\left|X\right|\right)} \leq A$, for some $A
> 0$. Naturally $X$ is assumed to be independent of the additive noise
$N$.

In the literature, multiple instances of such channel models were
investigated by making variations to the Shannon setup in the
following aspects:
\begin{itemize}
\item \underline {\em The input-output relationship $f(\cdot)$}: While
  Shannon considered a deterministic linear input-output relationship,
  many studies assumed a non-deterministic relationship~\cite{IA01,
    MK04, Nuriyev2005, Chung06} or generally a non-linear
  deterministic one~\cite{fahsj}.
 
\item \underline {\em The input constraint or cost function
  $\mathcal{C}(\cdot)$}: One of the main reasons of the popularity of
  the second moment constraint $\E{X^2}$ --which corresponds to a cost
  function $\mathcal{C}(x) = x^2$, is that it represents the average
  power of the discrete time transmitted signal which is equal to the
  average power of the corresponding white continuous process assuming
  that the transmitted signals are square integrable. Nevertheless,
  other input constraints were studied starting with
  Smith~\cite{SMITH71} who considered peak power constraints and a
  combination of peak and average power constraints.
  More recently, the capacity of Gaussian Channels with duty cycle and
  average power constraints was studied in~\cite{Zhang2011} .

\item \underline {\em The noise distribution}: Though Gaussian
  statistics of the noise can be motivated by the Central Limit
  Theorem (CLT), it also has an appealing property of being the worst
  case noise from an entropy perspective among finite second moment
  Random Variables (RV)s. Nevertheless, Non-Gaussian average power
  constrained communication channels have some applications and their
  channel capacities were investigated under a general setup in the
  work of Das~\cite{Das} where the noise is assumed to have a finite
  second moment, a condition that was not imposed on the non-Gaussian
  noise distributions in~\cite{Fahs2}.

\item \underline {\em Combinations of more than one aspect} were also
  considered in the literature.  Smith~\cite{SMITH71} extended his
  capacity results for the peak power constrained Gaussian channel to
  non-Gaussian ones where the noise statistics are Gaussian
  like. Later, Tchamkerten~\cite{Aslan} considered a scalar additive
  channel whose input is amplitude constrained and for which the
  additive noise is assumed to satisfy some general properties however
  not necessarily having a finite second moment. Lately, Fahs and
  Abou-Faycal~\cite{IA10,fahsj} investigated non-linear Gaussian
  channels under a general setup of input constraints such as even
  moments, compact support constraints and a combination of both
  types. Finally, channel capacity under Fractional Order Moments
  (FOM) of the form $\E{ |X|^{r} } \leq A$, for some $A > 0$, $r > 0$
  was characterized under a symmetric alpha-stable additive
  noise~\cite{Fahs12} or when the noise has two components, an
  alpha-stable component and a Gaussian one~\cite{Fahs14}.
\end{itemize} 

Nearly, for all the cited models above, and whenever the noise
Probability Density Function (PDF) is assumed to have an analytical
extendability property, the optimal input is proven to be of a
discrete nature and in most cases with a finite number of mass
points. Additionally, channel capacity could not be written in
closed-form. In this sense, the linear AWGN channel and some
``equivalent'' channels~\cite{fahsj} seem to be an exception, along
with a few channel models such as the additive exponential noise
channel under a mean constraint with non-negative inputs~\cite{Anan96}
and recently the Cauchy channel under a logarithmic
constraint~\cite{Fahs14-1}. For these channel models, the optimal
input distribution is found to be of the same nature of the noise and
capacity is described in closed-form.

One is tempted to study whether there is a general relation between
the input-output function $f(\cdot)$, the input cost function
$\mathcal{C}(\cdot)$, and the noise PDF $p_N(\cdot)$ that governs the
type of the capacity-achieving input. In this work, we conduct this
study for general types of the considered channel whereby ``general''
we mean the input-output relationship may not be linear but required
nevertheless to satisfy some rather mild conditions. Additionally,
instead of formulating the problem in terms of the average power
constraint or other moments constraints, we use generic input cost
functions that are also required to satisfy some technical
conditions. We emphasize that our results cover all cost functions
which are ``super-logarithmic'' which is a rather very large set. When
it comes to the noise statistics, the noise is assumed to be
absolutely continuous with respect to the Lebesgue measure with
positive and continuous PDFs that are with or without monotonic tails
and have a finite logarithmic-type of moments. Two conditions are
however imposed on the noise PDF and are subsequently presented. The
first guarantees the finiteness of the noise differential entropy. The
second concerns the tail behavior of a lower envelope to the noise
PDF. These two conditions are ``easily satisfied'' such as whenever
the PDF has a dominant exponential or a dominant polynomial component.
Despite the apparent long list of requirements, we emphasize that the
considered functions $f(\cdot)$, input costs $\mathcal{C}(\cdot)$ and
noise PDFs cover the vast majority of the known models found in the
literature.

Our main contributions are fourfold:
\begin{itemize}
\item[1-] Our study provides new capacity results for a multitude of
  communication channels. We showcase some of them:
  \begin{itemize}
  \item {\em Gaussian mixtures\/} and {\em generalized Gaussian\/}
    noise distributions are commonly used in the
    literature~\cite{blu,amir,fiorina2006,BSF08}, however no previous
    channel capacity studies were conducted for these types of noise
    models. The application of our results to such channels is
    presented in Sections~\ref{gaumix} and~\ref{gengau} respectively.
  \item Many communication channels are suitably modeled as {\em
      impulsive channels\/} where the statistics of the noise have,
    for example, an {\em alpha-stable\/} distribution or a composite
    {\em alpha-stable\/} plus {\em Gaussian\/} such as telephone
    noise~\cite{stuck}, audio noise signals~\cite{geo} and Multiple
    Access Interference (MAI)~\cite{sousa92,jacek,Win}. More
    recently~\cite{Beaulieu,marcel,ElGhannoudi,raj}, the performance
    of new receivers, mitigation and diversity techniques were
    investigated when such impulsive statistics were used as models of
    additive noise in MAI networks. In Section~\ref{asm}, we
    characterize and compute for a generic cost constraint
    $\mathcal{C}(\cdot)$ the channel capacity for the alpha-stable and
    the composite noise channels: {\em alpha-stable\/} plus {\em
      Gaussian\/}, which is commonly referred to as the Middleton class
    B when the stable component is symmetric~\cite{Middle86,Kim98}.
  \end{itemize}
\item[2-] The results stated in Theorems~\ref{th0} and~\ref{th00}
  generalize those of Das~\cite{Das}, who made similar statements for
  a linear channel whenever $\mathcal{C}(x) = x^2$ and whenever the
  noise is restricted, among other things, to have a finite second
  moment. The generalization is one to generic possibly non-linear
  channels, generic cost functions and noise distributions. They are
  in line with all the previous channel-capacity results presented
  earlier. In addition, when having an analyticity property of the
  noise PDF, they recover literally most of the known discreteness
  results for cost constrained deterministic channels. These results
  are stated in Theorem~\ref{fmass}.
\item[3-] Our methodology also provides capacity results even when the
  input is subjected only to a support (such as a peak power
  constraint), or to a combination of support and cost constraints. In
  fact, the results stated in Theorem~\ref{nfmass} corroborate those
  found in~\cite{Aslan} for channels whose input is amplitude
  constrained and where the noise has a finite $r$-th moment
  constraint for some $r >0$. However, Theorem~\ref{nfmass} is in some
  sense more general as it holds whenever the input has a compact
  support for {\em all\/} noise distributions that have a finite
  ``super-logarithmic'' moment.
\item[4-] When applied to monotonically-tailed noise PDFs for example,
  our main results - stated in Theorems~\ref{th0} and~\ref{th00},
  imply that whenever $\mathcal{C}(x) = \omega \left( \ln \left[
      \frac{1}{p_N \left( f(x) \right)} \right] \right)$\footnote{In
    this work, we say that $f(x) = \omega \left(g(x)\right)$ if and
    only if $\forall \, \kappa > 0, \exists \, c > 0$ such that
    $|f(x)| \geq \kappa |g(x)|, \forall |x| \geq c$. Equivalently, we
    say that $g(x) = o\left(f(x)\right)$. We say that $f(x) = \Omega
    \left(g(x)\right)$ if and only if $\exists \, \kappa > 0, c > 0$
    such that $|f(x)| \geq \kappa |g(x)|, \forall |x| \geq c$.
    Equivalently, we say that $g(x) = O\left(f(x)\right)$. We say that
    $f(x) = \Theta \left(g(x)\right)$ if and only if $f(x) = O
    \left(g(x)\right)$ and $f(x) = \Omega \left(g(x)\right)$.}, the
  support\footnote{We define the support of a RV as being the set of
    its points of increase i.e. $\{x \in \Reals: \Pr(x-\eta < X <
    x+\eta)>0 \text{ for all } \eta>0 \}$.}  of the capacity-achieving
  input is necessarily bounded. In addition, we state and prove a
  converse statement that states that whenever $\mathcal{C}(x) = o
  \left( \ln \left[ \frac{1}{p_N \left( f(x) \right)} \right]
  \right)$, the optimal input is necessarily unbounded. These results
  state that --for monotonically noise PDFs, there exists a threshold
  growth rate for the cost function which constitutes the transition
  between bounded and unbounded optimal inputs. Indeed, for an optimal
  input to be unbounded, a ``necessary condition'' for the cost
  function is to be at most $\Theta \left( \ln \left[ \frac{1}{p_N
        \left( f(x) \right)} \right]\right)$. This condition is
  satisfied by the Gaussian channel under an average power constraint,
  the exponential channel under a mean moment constraint and the
  Cauchy channel under a logarithmic constraint for which a Gaussian,
  exponentially tailed and a Cauchy input are respectively
  optimal~\cite{fahsj,Anan96,Fahs14-1}.
\end{itemize} 

On a related note, we argue that this study does provide insights on
what is a suitable measure of signal strength.
Though this question is not crucial when the additive noise has a
finite second moment due to the natural power measure provided by the
second moment, it seems of great importance when dealing with heavy
tailed noise distributions having infinite second moments. For these
types of channels, since the noise has an infinite second moment,
using the second moment as a measure of signal strength is absurd for
evaluating the Signal-to-Noise Ratio (SNR) for example. In fact, for
heavy-tailed noise models, and more specifically for the alpha-stable
class, a general theory of stable signal processing based on
Fractional Lower Order Moments (FLOM) ($\cost{x} = |x|^r$, $0< r <2$)
was presented in~\cite{Shao}; The ``stable theory'' was in accordance
with the fact that second order methods and linear estimation theory
were no longer suitable for infinite variance additive noise channels
and new criteria based on the dispersion of alpha-stable RVs and FLOM
were investigated. The stable theory was also used in the treatment of
various detection and estimation problems~\cite{tsa,man}, and the
performance of optimum receivers were investigated in~\cite{tsi}.

However, one can argue that moments of the form $\E{|X|^r}$, $r >0$ do
not provide a suitable strength measure for variables having
heavy tailed distributions, simply because no single value of $r$ can
be found appropriate~\cite{gonza2}. This conclusion is further
supported when making the following reasoning on the additive noise
channels considered in this work:
\begin{itemize}
\item Let $\E{\mathcal{C}_0(|X|)}$ be a measure of the average signal
  strength where $\mathcal{C}_0(|x|)$ is some positive, lower
  semi-continuous, non-decreasing function of $|x|$ and let $p_N(x)$
  be the noise PDF which is assumed to have a monotonic tail.
\item Whenever $\mathcal{C}_0(|x|) = \omega \left( \ln \left[
      \frac{1}{p_N(f(x))} \right] \right)$, one can always find a cost
  function $\mathcal{C}(|x|)$ such that $\mathcal{C}(|x|)$ is both
  $\omega \left( \ln \left[ \frac{1}{p_N(f(x))} \right] \right)$ and
  $o \left( \mathcal{C}_0(|x|) \right)$.
\item Now, since $\mathcal{C}(|x|) = \omega \left( \ln \left[
      \frac{1}{p_N(f(x))} \right] \right)$, the channel capacity under
  an input constraint of the form $\E{\mathcal{C}(|X|)} \leq A$, $A
  >0$ is achieved by a bounded input by virtue of
  Theorem~\ref{th0}. On the other hand, since $\mathcal{C}(|x|) = o
  \left( \mathcal{C}_0(|x|) \right)$, then there exists a distribution
  function satisfying the cost constraint with ``signal strength''
  $\E{ \mathcal{C}_0(|X|) }$ equal to $\infty$.
\item Hence, in the input space of distribution functions, there exist
  distributions having possibly infinite strength while the capacity
  is achieved by a distribution which has a finite one since
  its support is bounded.
\item This non-intuitive conclusion is only possible under the choice
  of a strength measure that is $\omega \left( \ln \left[
      \frac{1}{p_N(f(x))} \right] \right)$.
\end{itemize}

By this reasoning, suitable signal strength measures should be at most
$\Theta \left( \ln \left[ \frac{1}{p_N(f(x))} \right] \right)$. Said
differently, depending on the noise, measures of the form $\Theta
\left( \ln \left[ \frac{1}{p_N(f(x))} \right] \right)$ are more
appropriate. This boils down to $\Theta(f^2(x))$ under the Gaussian
noise
and to $\Theta \left( \ln \left[ f(x) \right] \right)$ for
polynomially-tailed additive noise.  The latter condition comes in
accordance with the work of Gonzalez et al.~\cite{gonza2} who
presented a new approach for dealing with heavy-tailed noise
environments. After presenting the shortcomings of the FLOM approach,
they presented a ``general'' unit of strength-measure based on
logarithmic moments where they motivated its usage within the
framework of estimation and filtering under impulsive noise.

The remainder of this paper is organized as
follows. Section~\ref{sc:generic_ch} presents the generic channel
model along with all the assumptions made in our study. Preliminary
lemmas concerning lower and upper bounds on some quantities of
interest are stated and proven in Section~\ref{sc:Preliminaries}. In
Section~\ref{KKT}, we discuss the Karush-Kuhn-Tucker (KKT) theorem and
in Section~\ref{MR} the main results of this paper are presented as
Theorems 1, 2, 3 and 4. Examples and numerical evaluations of channel
capacity and its achieving input distributions are presented in
Section~\ref{exnum} where the application of the four theorems is
explored for Gaussian mixtures, Generalized Gaussians and impulsive
noise. Finally, Section~\ref{conc} concludes the paper.


\section{A generic channel model}
\label{sc:generic_ch}

We consider a generic memoryless real discrete-time noisy
communication channel where the noise is additive and where the input
and output are possibly non-linearly related as follows:
\begin{equation}
  Y_{i} = f(X_i) + N_{i},
  \label{eq:generic_ch}
\end{equation}
where $i$ is the time index. We denote by $Y_{i} \in \Reals$ the
channel output at time $i$. The input at time $i$ is denoted $X_{i}$
and is assumed to have an alphabet $\set{X} \subseteq \Reals$. The
channel's input is distorted according to the deterministic and
possibly non-linear function $f(x)$.
Additionally, the communication channel is subjected to an additive
noise process that is independent of the input. The variables
$\{N_{i}\}_i$ are also assumed to be Independent and Identically
Distributed (IID) RVs.

Finally, we subject the input to an average cost constraint of the
form: $\E{\cost{X_{i}}} \leq A$, for some $A \in \Reals^{+*}$ where
$\mathcal{C}(\cdot)$ is some cost function:
\begin{equation*}
  \mathcal{C}: \Reals^+ \longrightarrow \Reals.
\end{equation*}

Accordingly, we define for $A > 0$
\begin{equation}
  \set{P}_A = \Big\{ \text{ Probability distributions } F \text{ of } X :\int \cost{x}\,dF(x) \leq A \Big\},
  \label{consdef100}
\end{equation}
the set of all distribution functions satisfying the average cost
constraint.

Given that the channel model is stationary and memoryless, the
capacity-achieving statistics of ${X_i}$ are also memoryless (IID),
therefore we suppress the time index and write
\begin{equation}
  Y = f(X) + N,
  \label{eq:generic_ch1}
\end{equation}
where the noise is absolutely continuous with respect to the Lebesgue
measure and is assumed to have a PDF $p_{N}(\cdot)$. This implies that
the channel transition probability density function is given by
\begin{equation}
  p_{Y|X}(y|x) = p_N (y-f(x)), \quad y \in \Reals, \,\,x \in \set{X}.
  \label{eq:py_x}
\end{equation}

We characterize the tail behavior of $p_N (\cdot)$ by considering the
following positive functions which are non-increasing for $x \geq 0$
and non-decreasing for $x < 0$:
\begin{equation*}
  \Tl{x} = \left\{ \begin{array}{ll}
      \displaystyle \inf_{0 \leq t \leq x}p_N(t) & \quad x \geq 0 \\ 
      \displaystyle \inf_{x \leq t \leq 0}p_N(t) & \quad x < 0,
    \end{array} \right.
  \qquad 
  \Tu{x} = \left\{ \begin{array}{ll}
      \displaystyle \sup_{ t \geq x} \, \, p_N(t) & \quad x \geq 0\\ 
      \displaystyle  \sup_{t \leq x} \, \, p_N(t) & \quad x < 0.
    \end{array} \right.
\end{equation*}

Considering the tail behavior of $\Tl{x}$ and $\Tu{x}$ instead of $p_N
(x)$ allows us to include in our analysis PDFs which do not possess a
monotonic tail. For those that do, $p_N(x)$, $\Tl{x}$ and $\Tu{x}$
will be identical for large values of $|x|$.

The main results of this work are based on relating the tail behavior
of $\mathcal{C}(\cdot)$ to that of $\Tl{\cdot}$ and $\Tu{\cdot}$ in
order to characterize the capacity-achieving input distributions of
channel~(\ref{eq:generic_ch1}). More explicitly we prove that,
whenever $\cost{x} = \omega\left( \ln \left[ \frac{1}{\Tl{f(x)}}
  \right] \right)$, the optimal input has necessarily a bounded
support. Furthermore, we prove a converse statement: whenever
$\displaystyle \cost{x} = o \left(\ln \left[ \frac{1}{\Tu{f(x)}}
  \right] \right)$, the capacity-achieving input is not bounded.

When the noise PDF has a monotonic tail, our results infer that cost
functions which are $\Theta \left(\ln \left[ \frac{1}{p_N(f(x))}
  \right] \right)$ form somehow a ``transition'' between bounded and
unbounded optimal inputs. For example, whenever the noise is Gaussian,
the ``transitional'' cost is of the form
$\Theta\left(f^2(x)\right)$. The discreteness --and hence the
finiteness of the number of mass points of the optimal input in the
bounded case-- is a direct consequence of the analyticity properties
of $p_N(\cdot)$ and $\mathcal{C}(\cdot)$ whenever these properties
exist.

\subsection{Assumptions} 
\label{conditions}

In this work, we make the following assumptions:
\begin{itemize}
\item[$\bullet$] \underline{The function $f(\cdot)$:}
  \begin{itemize}
  \item[C1-] The function is continuous.
  \item[C2-] The absolute value of the function $\left| f(\cdot)
    \right|$ is a non-decreasing function of $|x|$ and $\left| f(x)
    \right| \rightarrow +\infty$ as $|x| \rightarrow +\infty$.
  \end{itemize}
  
\item[$\bullet$] \underline{The cost function $\mathcal{C}(\cdot)$:}
  \begin{itemize}
  \item[C3-] The cost function is lower semi-continuous and
    non-decreasing. Without Loss of Generality (WLOG) we assume that
    $\mathcal{C}(0) = 0$: if it were not, define $\mathcal{C}_{0}(|x|)
    = \mathcal{C}(|x|) - \mathcal{C}(0)$ and adjust the input space
    under the cost $\mathcal{C}_0(|x|)$ to
    $\mathcal{P}_{A-\mathcal{C}(0)}$. Note that necessarily $A -
    \mathcal{C}(0) \geq 0$.
  \item[C4-] $\cost{x} = \omega\left(\ln\left|f(x)\right|\right)$.
  \end{itemize}

\item[$\bullet$] \underline{The noise PDF $p_{N}(\cdot)$:}
  \begin{itemize}
  \item[C5-] The PDF is positive and continuous on $\Reals$. Note that
    this automatically implies that $p_N(\cdot)$ is upper bounded.
  \item[C6-] There exits a non-decreasing function
    \begin{equation*}
      \mathcal{C}_N: \Reals^+ \longrightarrow \Reals
    \end{equation*}
    such that $ \mathcal{C}_N\left(|x|\right) = \omega \left( \ln |x|
    \right)$, and
    \begin{equation*}
      \Ep{N}{ \mathcal{C}_N\left(\left|N\right|\right) } = L_N < \infty.
    \end{equation*}
    This necessarily implies that $ \Ep{N}{\ln \left( 1 + |N| \right)}
    < \infty$. Note that, for example, the above condition holds true
    for any noise PDF whose tail is faster than $\frac{1}{x\left(\ln
        x\right)^3}$.
  \end{itemize}

  Since from an information theoretic perspective, the general channel
  model~(\ref{eq:generic_ch}) is invariant with respect to output
  scaling, we consider WLOG that the noise PDF is less than ``1'' for
  technical reasons. Furthermore, the boundedness of $p_N(\cdot)$
  along with the fact that it has a finite logarithmic moment insure
  that its differential entropy exists and is finite $h(N) < \infty$
  (see~\cite[Proposition 1]{rioul2011}).

  Restrictions C1 to C6 are ``technical'' in the sense that they
  represent sufficient conditions for the existence of a solution to
  the capacity problem as defined in~\cite{SMITH71} and enables the
  formulation of the Karush Kuhn Tucker (KKT) conditions as being
  necessary and sufficient for optimality of the input probability
  distribution.

\item[$\bullet$] \underline{The lower and upper bounds $\Tl{\cdot}$
    and $\Tu{\cdot}$:}

  Note that by definition, $0 < \Tl{x} \leq p_N(x) \leq \Tu{x} \leq 1$
  for all $x \in \Reals$. We assume that $\Tl{\cdot}$ and $\Tu{\cdot}$
  satisfy the following properties:
  \begin{itemize}
  \item[C7-] The function $L(x) = \ln \left[ \frac{1}{\Tl{x}}
    \right]$ which is positive, non-decreasing for $x \geq 0$ and
    non-increasing in $x < 0$, satisfies the following inequality:
    \begin{equation}
      L(x+y) \leq \kappa_{\text{l}} \left(L(x) + L(y)\right),
      \label{ident}
    \end{equation} 
    for some positive constant $\kappa_{\text{l}}$, whenever $|x|$,
    $|y|$ are sufficiently large.
  \end{itemize}

  We note that functions that satisfy condition C7 define a convex
  set. In fact, let $f(x)$, $g(x)$ be two positive, non-decreasing
  functions on $\Reals^+$ non-increasing on $\Reals^{-*}$. Let $\alpha
  \in [0,1]$ and define $h = \alpha f + (1-\alpha) g$. The function
  $h(x)$ is positive, having the same monotonic properties. Then,
  whenever there exists $\kappa_{\text{f}}$ and $\kappa_{\text{g}} >
  0$ for which $f$ and $g$ satisfy condition C7, we have
  \begin{equation*}
    h(x+y) = \alpha f(x+y) + (1-\alpha)g(x+y) \leq \kappa_{\text{h}}(h(x) + h(y)),
  \end{equation*}
  where $ \kappa_{\text{h}} = \max\{ \kappa_{\text{f}};
  \kappa_{\text{g}}\} > 0$.

  We clarify that condition C7 is for example satisfied by all noise
  distribution functions where $\Tl{x}$ is any linear combinations of:
  \begin{equation*}
    \Tl{x} = \Theta\left(s(x)e^{r(x)}\right) \qquad 
    \Tl{x} = \Theta\left(\frac{s(x)}{r(x)}\right),
  \end{equation*}
  where 
  \begin{equation*}
    r(x) = |x|^{a}\underbrace{\log \dots \log(|x|)}_{\beta \text{ times}}, \qquad s(x) = |x|^{a{'}}
    \underbrace{\log...\log(|x|)}_{\beta^{'} \text{ times}},
  \end{equation*}
  and where the parameters $a$, $a{'}$ $\in \Reals^+$, and $\beta$,
  $\beta^{'} \in \Naturals$, chosen so that $\Tl{x}$ is positive, its
  total integral is no greater than one, and conserves its monotonic
  behavior\footnote{The values $\beta= 0$ and $\beta^{'} = 0$ imply
    that respectively $r(x)$ and $s(x)$ have no logarithmic
    component.}. The fact that these two general types satisfy
  condition C7 is based on the following basic
  identities~\cite{vonb}:
  \begin{itemize}
  \item[$\bullet$] For all $x$, $y$ and $r \in \Reals$,
    \begin{equation*}
      |x+y|^r \leq \max\{1;2^{r-1}\}\left(|x|^r + |y|^r\right).
    \end{equation*}
  \item[$\bullet$] For any $x_0 > 0$, there exist $y_0 >0$ such that
    \begin{equation*}
      |x|+|y| \leq |xy|^p, \text{ for some $p > 1$ whenever $|x| > x_0$, $|y| > y_0$.}
    \end{equation*}
  \end{itemize}
  

  Finally, we also assume that
  \begin{itemize}
  \item[C8-] The integral $\displaystyle - \int_{-\infty}^{+\infty}
    \Tu{x} \ln \Tl{x}\,dx$ exists and is finite.
  \end{itemize}
  
  Note that whenever the tail of $p_N(\cdot)$ is monotone, condition
  C8 is not necessary and boils down to saying that noise
  differential entropy is finite which is a byproduct of properties C5
  and C6 of the noise PDF.
  
  When it comes to conditions C5 through C8 --and specifically C7
  and C8--, they are satisfied by a rather large class of noise
  probability functions that includes most of the known probability
  models such as Gaussian, generalized Gaussian, generalized t,
  alpha-stable distributions and all of their possible mixtures.
\end{itemize}


\section{Preliminaries}
\label{sc:Preliminaries}

In this section we establish some preliminary results that are needed
in subsequent sections: we derive lower and upper bounds on the output
probability and a quantity of interest presented hereafter.

We start by noting that for channel~(\ref{eq:generic_ch1}), the
existence of a positive, continuous transition PDF such as
in~(\ref{eq:py_x}), implies the existence for any input distribution
$F$ of an induced output probability density function $p_Y(y)= p(y;F)$
which is also continuous (hence upper-bounded)~\cite{fahsj} and is
given by:
\begin{equation}
  p_Y(y;F) = p(y;F) = \int p_N(y-f(x))\,dF_X(x) \leq 1.
  \label{outnorm}
\end{equation} 

Furthermore, equation~(\ref{outnorm}) along with the fact that
$f(\cdot)$ is continuous insures that the property that $p_{N}(\cdot)$
is bounded away from zero on compact subsets of $\Reals$ is conserved
as well for $p_Y(y;F)$. This in turns implies that $p_Y(\cdot)$ is
also positive on $\Reals$.

\subsection{Bounds on $p(y;F)$}

In what follows, we derive upper and lower bounds on the output
probability distribution induced by an input distribution $F$.

\begin{lemma}
  \label{lemlow}
  Let $y_0 > 0$ be sufficiently large. For an input distribution $F$,
  the PDF $p(y;F)$ of the output of channel~(\ref{eq:generic_ch1}) is
  lower bounded by
  \begin{equation*}
    p(y;F)  \geq \left\{ \begin{array}{ll}
        \displaystyle \frac{\Tl{y-y_0}}{2} & \quad y \leq -y_0\\ 
        \displaystyle \frac{\Tl{y+y_0}}{2} & \quad y \geq \,\,\,\, y_0,
      \end{array} \right.
  \end{equation*}
\end{lemma}

\begin{proof}
  Given an input probability distribution $F$, we define the
  following:
  \begin{itemize}
  \item[-] We denote by $d_F$ a positive constant such that $\Pr(|X|
    \leq d_F) \geq \frac{1}{2}$.
  \item[-] We denote by $f_{\max} = \sup_{|x| \leq d_F}|f(x)|$, the
    existence of which is guaranteed by the assumption that $f(\cdot)$
    is continuous on $\Reals$.
  \end{itemize}

  Let $y_0 > f_{\text{max}}$. In what follows, we only present in
  detail the case $y \geq y_0$ as the proof in the other range follows
  similar steps.
  \begin{align}
    p_Y(y;F) \geq& \int_{x : |x| \leq d_F} \, p_N \bigl( y-f(x) \bigr) \, dF(x) \nonumber\\
    \geq&  \int_{x : |x| \leq d_F} \,  \Tl{y-f(x)} \, dF(x) \label{tailbeh}\\
    \geq& \,\frac{1}{2}\Tl{y+f_{\max}} \geq \,\frac{1}{2}\Tl{y+y_0},  \label{0dec}
  \end{align}
  where equation~(\ref{tailbeh}) is due to the fact that $\Tl{\cdot}$
  is a lower bound on $p_{N}(\cdot)$ by definition and
  inequalities~(\ref{0dec}) are justified since $\Tl{\cdot}$ is
  non-increasing on the considered interval.
\end{proof}

We also derive an upper bound on the output law whenever the input is
bounded within $[-B,B]$ for some $B > 0$:
\begin{lemma}
  \label{lemhigh}
  For an input distribution $F$ that has a bounded support within
  $[-B,B]$ for some $B > 0$, the PDF $p(y;F)$ of the output of
  channel~(\ref{eq:generic_ch1}) is upper bounded by
  \begin{equation*}
    p(y;F) \leq \left\{ \begin{array}{ll}
        \displaystyle \Tu{y+y_0^{B}} & \quad y \leq -y^{B}_0\\ 
        \displaystyle \Tu{y-y_0^{B}} & \quad y \geq \,\,\,\, y^{B}_0,
      \end{array} \right.
  \end{equation*}
  for any large-enough $y^{B}_0$.
\end{lemma}

\begin{proof}
  Let $f_{\max}^{B} = \sup_{[-B;B]} \left| f(x) \right|$, the
  existence of which is guaranteed by the fact that $f(\cdot)$ is
  continuous on $\Reals$. Also let $y^{B}_0 \geq f_{\max}^{B}$.  For
  $y \geq y^{B}_0$, since $\Tu{\cdot}$ is an upper bound on
  $p_N(\cdot)$, we have,
  \begin{eqnarray}
    p(y;F) &=& \int p_N (y-f(x))\,dF(x) \nonumber\\
    &=& \int_{-B}^{B} p_N (y-f(x))\,dF(x) \label{onlyposs}\\
    &\leq&\int_{-B}^{B} \Tu{y-f(x)}\,dF(x) \nonumber \\
    &\leq& \Tu{y-f_{\max}^{B}} \leq \Tu{y-y_0^{B}}\label{uselat}
  \end{eqnarray}
  where equations~(\ref{uselat}) are due to the fact that $\Tu{x}$ is
  non-increasing on the positive semi-axis. A similar derivation
  yields the result for $y \leq -y_0^{B}$.
\end{proof}

We emphasize that this upper bound on $p(y;F)$ is only possible under
the assumption that the support of $F$ is bounded (as seen in
equation~(\ref{onlyposs})).

\subsection{Bounds on $i(x;F)$}

In this section we analyze the function of interest
\begin{align}
  i(x;F) & = - \int_{-\infty}^{+\infty} p_N(y-x) \ln p_Y(y;F)\,dy \label{eq:i} \\
  & = - \int_{-\infty}^{+\infty} p_N(y) \ln p_Y(y+x;F)\,dy. \nonumber
\end{align}

\begin{lemma}
  \label{le:iUB}
  For any probability distribution $F$, 
  \begin{equation*}
    i(x;F) = O \left( \ln \left[ \frac{1}{\Tl{x}} \right] \right).
  \end{equation*}
\end{lemma}

\begin{proof}
  Consider a large-enough $y_0$ so that Lemma~\ref{lemlow} holds, and
  let $x$ be such that $x>y_0$. For a probability distribution $F$ on
  the input we compute,
  \begin{equation*}
    i(x;F) = - \int_{-\infty}^{+\infty} p_N(y) \ln p_Y(y+x;F)\,dy = I_1 + I_2 + I_3,
  \end{equation*}
  where the interval of integration is divided into three
  sub-intervals: $(-\infty,-x-y_0)$, $[-x-y_0, y_0]$, $(y_0,+\infty)$.
 
  We study the growth rate in $x$ of the integral terms $I_1$, $I_2$
  and $I_3$ function of the rate of decay of $\Tl{\cdot}$.

  Using Lemma~\ref{lemlow},
  \begin{align}
    I_1 =& -\int_{-\infty}^{-x-y_0} p_N(y) \ln p_Y(y+x;F)\,dy \nonumber \\
    \leq & - \int_{-\infty}^{-x-y_0} p_N(y) \ln \left[ \frac{\Tl{y + x - y_0}}{2} \right] \, dy
    = \int_{-\infty}^{-x-y_0} p_N(y) \ln \left[ \frac{2}{\Tl{y +x - y_0}}\right] \, dy \nonumber \\
    \leq&  \, \ln 2 + \kappa \int_{-\infty}^{-x-y_0} p_N(y) \left(\ln \left[\frac{1}{\Tl{y}} \right]+ 
      \ln \left[\frac{1}{\Tl{x}} \right] + \ln \left[\frac{1}{\Tl{-y_0}} \right]\right)\,dy \label{justn1}\\
    \leq&  \, \ln 2 + \kappa \, \ln \left[\frac{1}{\Tl{-y_0}} \right]+ \kappa \,\ln \left[\frac{1}{\Tl{x}}  \right]
    + \kappa \,\int_{-\infty}^{+\infty} p_N(y)\ln \left[\frac{1}{\Tl{y}} \right] dy\label{justn2}\\
    \leq & \, 2 \kappa \, \ln \left[ \frac{1}{\Tl{x}}\right] \nonumber,
  \end{align}
  for some positive $\kappa$ and for $x>y_0$
  large-enough. Equation~(\ref{justn1}) is due to property C7 since
  both $x$ and $y_0$ are large enough and so is $|y|$. The integral
  term in~(\ref{justn2}) is finite by property C8 and the last
  equation is valid since $\ln \left[ \frac{1}{\Tl{x}} \right]$, which
  is positive, is increasing to +$\infty$.

  Similarly,
  \begin{align}
    I_3 =& -\int_{y_0}^{\infty} p_N(y) \ln p_Y(y+x;F)\,dy \nonumber \\
    \leq & - \int_{y_0}^{\infty} p_N(y) \ln \left[ \frac{\Tl{y + x + y_0}}{2} \right] \, dy
    = \int_{y_0}^{\infty} p_N(y) \ln \left[ \frac{2}{\Tl{y +x + y_0}}\right] \, dy \nonumber \\
    \leq&  \, \ln 2 + \kappa \int_{y_0}^{\infty} p_N(y) \left(\ln \left[\frac{1}{\Tl{y}} \right]+ 
      \ln \left[\frac{1}{\Tl{x}} \right] + \ln \left[\frac{1}{\Tl{y_0}} \right]\right)\,dy \nonumber\\
    \leq & \, 2 \kappa \, \ln \left[ \frac{1}{\Tl{x}}\right] \nonumber,
  \end{align}

  As for $I_2$,
  \begin{eqnarray}
    I_2 &=& -\int_{-x-y_0}^{y_0} p_N(y) \ln p_Y(y+x;F)\,dy \nonumber\\
    &=& -\int_{-x-y_0}^{-x+y_0} p_N(y) \ln p_Y(y+x;F)\,dy -\int_{-x+y_0}^{y_0} p_N(y) \ln p_Y(y+x;F)\,dy \nonumber\\
    &\leq& \sup_{|y| \leq y_0} \ln \left[ \frac{1}{p_Y(y;F)} \right] + \int_{-x+y_0}^{y_0} p_N(y) \ln \left[\frac{2}{\Tl{y + x + y_0}}\right] dy \nonumber\\
    &\leq& \sup_{|y| \leq y_0} \ln \left[ \frac{1}{p_Y(y;F)} \right]+ \ln 2 + \ln \left[ \frac{1}{\Tl{x + 2 y_0}}\right] \label{proptl}\\
    &\leq&  \sup_{|y| \leq y_0} \ln \left[ \frac{1}{p_Y(y;F)} \right] + \ln 2 + \kappa \ln \left[ \frac{1}{\Tl{x}} \right] + \kappa
    \ln \left[ \frac{1}{\Tl{2 y_0}} \right]\label{proptl1}\\
    &\leq& 2 \, \kappa  \, \ln \left[ \frac{1}{\Tl{x}}  \right]. \nonumber
  \end{eqnarray} 
  The supremum is finite since it is taken over a compact set where
  $p_Y(y)$ (which is less than one) is positive, continuous and hence
  positively lower bounded. Equation~(\ref{proptl}) is due to the fact
  that $\ln \left[ \frac{1}{\Tl{\cdot}} \right]$ is non-decreasing on
  the positive axis, equation~(\ref{proptl1}) is given by property
  C7 since both $x$ and $y_0$ are large enough and the last equation
  is justified since $\ln \left[ \frac{1}{\Tl{x}} \right]$ is
  increasing to $+\infty$ as $|x| \rightarrow +\infty$.

  A similar procedure can be adopted to prove this result when $x
  \rightarrow -\infty$ by adjusting the intervals of integration to
  the following: $(-\infty, -y_0)$, $[-y_0, -x+y_0]$, $(-x+y_0,
  +\infty)$ where $x < -y_0$ such that $|x|$ is large enough. This
  would imply that for any probability distribution $F$, $i(x;F) = I_1
  + I_2 + I_3 = O \left(\ln \left[ \frac{1}{\Tl{x}} \right] \right)$.
\end{proof}

We also derive a lower bound whenever the input is bounded within
$[-B,B]$ for some $B > 0$:
\begin{lemma}
  \label{le:iLB}
  For an input distribution $F$ that has a bounded support within
  $[-B,B]$ for some $B > 0$,
  \begin{equation*}
    i(x;F) = \Omega \left( \ln \left[ \frac{1}{\Tu{x}} \right] \right).
  \end{equation*}
\end{lemma}

\begin{proof}
  We proceed in a manner akin to the proof of Lemma~\ref{le:iUB}: For
  an input distribution $F$ that has a bounded support within $[-B,B]$
  for some $B > 0$, we consider a large-enough $y^B_0$ so that
  Lemma~\ref{lemhigh} holds, and let $x$ be such that $x>y^B_0$.
  \begin{eqnarray}
    i(x;F) & = & - \int_{-\infty}^{\infty} p_N(y) \ln p(y+x;F)\,dy \nonumber\\
    & \geq & -\int_{y^B_0}^{+\infty} p_N(y) \ln p(y+x;F)\,dy\nonumber\\
    & \geq & \int_{y^B_0}^{+\infty} p_N(y) \ln \left[ \frac{1}{\Tu{y + x - y^{B}_0}} \right] dy \label{newlow}\\
    & \geq & \left(1 - F_N(y_0^B) \right) \ln \left[ \frac{1}{\Tu{x}}  \right] > 0. \label{boundw}
  \end{eqnarray}
  
  In order to write equation~(\ref{newlow}) we use the upper bound in
  Lemma~\ref{lemhigh}. Equation~(\ref{boundw}) is justified since $\ln
  \left[ \frac{1}{\Tu{\cdot}} \right]$ is non-decreasing on the
  non-negative semi-axis and the end result is positive since the
  support of $N$ is $\Reals$. A similar analysis may be conducted for
  the case when $x < -y_0^B < 0$.
\end{proof}


\section{The Karush-Kuhn-Tucker (KKT) theorem}
\label{KKT}

The capacity of channel~(\ref{eq:generic_ch}) is the supremum of the
mutual information $I(\cdot)$ between the input $X$ and output $Y$
over all input probability distributions $F$ that meet the constraint
$\mathcal{P}_A$:
\begin{equation}
  C = \sup_{F\in \mathcal{P}_A}\,I(F) = \sup_{F\in \mathcal{P}_A} \iint p_N\left(y-f(x)\right) \, 
  \ln \left[ \frac{p_N\left(y-f(x)\right)}{p(y;F)} \right] \,dy\,dF(x).
  \label{capprobl}
\end{equation}
 
Conditions C1 to C6 guarantee that this optimization problem is
well-defined and that its solution --the capacity-- is finite and is
achievable~\cite[Theorem 2]{FAF15}. Indeed, the conditions are
sufficient for $\mathcal{P}_A$ to be convex and compact~\cite[Theorem
3]{FAF15} and for $I(\cdot)$ to be concave and continuous (in
the weak sense~\cite[Sec.III.7]{ShiryBook})~\cite[Theorems
4,5]{FAF15}.

When dealing with constrained optimization problems, the Lagrangian
theorem~\cite{Luenb} is a useful tool as it transforms the problem to
an unconstrained one when some convexity conditions are satisfied by
the objective function and the constraints. In our problem these
conditions are satisfied as the mutual information is concave and the
cost is linear - and hence convex.  The theorem states that there
exists a non-negative parameter $\nu_A$ such that the optimization
problem~(\ref{capprobl}) can be written as:
\begin{align}
  C = \sup_{F\in \mathcal{P}_A}\,I(F) & = \sup_{F}\,\left\{ I(F) - \nu_A \bigl( \Ep{F}{\cost{X}} 
    - A \bigr) \right\} \label{eq:dual} \\
  & = I(F^*) - \nu_A \, \Ep{F^*}{\cost{X}} + \nu_A A, \nonumber
\end{align}
where the last equality is true since the solution is finite and
achievable by an optimal $F^*$. Furthermore,
\begin{equation*}
  \nu_A \, (\Ep{F^*}{\cost{X}} - A) = 0.
\end{equation*}

For every positive $A$, denote by $C(A)$ the capacity of the channel
under the constraint $F \in \mathcal{P}_A$, and consider the function
$C(A)$ for $A > 0$. The significance of the Lagrange parameter $\nu_A$
is addressed in the following Lemma.

\begin{lemma}
  \label{tang}
  Whenever for some positive $A$ the parameter $\nu(A) = 0$, then
  $C(A') = C(A)$ for all $A' \geq A$.
\end{lemma}

\begin{proof}
  We start by noting that the channel capacity $C(A)$ is a
  non-decreasing function of $A$, due to the fact $\mathcal{P}_A
  \subseteq \mathcal{P}_{A'}$, for $0 < A \leq A'$.
  %
  %
  %
  %
  Now assume that $\nu(A) = 0$ for some $A > 0$. For this value of
  $A$, equation~(\ref{eq:dual}) becomes
  \begin{equation*}
    C = \sup_{F\in \mathcal{P}_A}\,I(F) = \sup_{F}\,\left\{ I(F) - \nu_A \bigl( \Ep{F}{\cost{X}} 
      - A \bigr) \right\} = \sup_{F}\,I(F),
  \end{equation*}
  which is a maximal value over all probability distributions
  irrespective of the constraint. This observation along with the fact
  that $C(A)$ is non-decreasing establish the result.
\end{proof}

In our setup, a value of $\nu(A) = 0$ can be ruled out. Said
differently, the cost constraint in equation~(\ref{capprobl}) is
binding. The argument we make is similar to the one used
in~\cite{IA01}: we consider a family of input signals composed of $N$
discrete levels with equal probabilities at locations $\{1, L, L^2,
\cdots, L^{2^{N-2}}\}$. When $L$ increases, the probability of error
of a minimum probability of error receiver goes to zero, which implies
by Fano's inequality that the mutual information approaches $\ln
(N)$. Therefore, as $A \to \infty$, the achievable rates in our setup
are arbitrarily large and $C(A)$ increases to infinity; a fact that is
not possible if $\nu(A)$ were equal to zero for some $A$ by
Lemma~\ref{tang}. This conclusion is corroborated by the fact that the
capacity for general memoryless continuous-input, continuous-output
channels is achieved by a boundary input for unbounded input cost
functions~\cite{Arg12}.

Whenever weak (Gateaux) differentiability is guaranteed, one can
further write necessary and sufficient conditions on the maximum
achieving distribution; conditions that are commonly referred to as
the KKT conditions~\cite{Luenb}. More formal statements on the theory
of convex optimization are summarized in Appendix C in~\cite{fahsj}.
The KKT approach was used previously in many studies~\cite{Sha48_1,
  Gal68, Hirt88, IA01, Chung06, MK04,
  SMITH71,Aslan,Fahs,Zhang2011,Tsai2011,Fahs2} in order to solve the
capacity problem and for the purpose of this work, we follow similar
steps. We indeed prove in Appendix~\ref{ap:Prop_I} the weak
differentiability of $I(\cdot)$ at any optimal input $F^*$ and
proceeding as in~\cite{IA01}, we write the KKT conditions as being
necessary and sufficient conditions for the optimal input to
satisfy. These conditions state that an input RV $X^*$ with
probability distribution $F^*$ achieves the capacity $C$ of an average
cost constrained channel if and only if there exists $\nu \geq 0$ such
that,
\begin{equation}
  \nu (\cost{x} - A) + C + H + \int p_N \left(y-f(x)\right) \ln p(y;F^{*})\,dy 
  = \nu (\cost{x} - A) + C + H - i(f(x);F^*) \geq 0,
  \label{eq:KKT_1}
\end{equation}
for all $x$ in $\Reals$, with equality if $x$ is a point of increase
of $F^*$, and where $H$ is the entropy of the noise.

\section{Main results}
\label{MR}

\begin{theorem}
  \label{th0}
  Whenever $\cost{x} = \omega \left( \ln \left[ \frac{1}{\Tl{f(x)}}
    \right] \right)$, the support of the capacity-achieving input of
  channel~(\ref{eq:generic_ch1}) is compact.
\end{theorem}

\begin{proof}
  We consider the necessary and sufficient conditions of
  optimality~(\ref{eq:KKT_1}), and we study the behavior of the
  expression function of the variable $x$ as its magnitude goes to
  infinity.

  These conditions state that for the optimal input $X^*$,
  condition~(\ref{eq:KKT_1}) is satisfied with equality for any point
  of increase $x_0$ of the capacity-achieving distribution $F^* \in
  \mathcal{P}_A$. For such an $x_0$ we obtain,
  \begin{equation*}
    \nu (\cost{x_0} - A) + C + H = i(f(x_0);F^*).
  \end{equation*}
  If these points of increase  of $X^*$ take arbitrarily large values,
  $|f(x_o)|    \rightarrow   +\infty$    since   $|x_o|    \rightarrow
  +\infty$. Using Lemma~\ref{le:iUB}, $i(f(x_o);F) = O \left( \ln \left[
      \frac{1}{\Tl{f(x_o)}} \right] \right)$, and therefore
 \begin{equation*}
    \nu (\cost{x_0} - A) + C + H =  O \left( \ln \left[
        \frac{1}{\Tl{f(x_0)}} \right] \right),
  \end{equation*}
  which is a contradiction whenever $\cost{x}= \omega \left( \ln
    \left[ \frac{1}{\Tl{f(x)}} \right] \right)$ unless $\nu = 0$.
  This has been ruled out in Section~\ref{KKT}, which implies that the
  support of $X^*$ is bounded.  Finally, we note that the support is
  always closed, as its complement is open.  Therefore, $X^*$ is
  compactly supported.
\end{proof}

\subsection{A converse theorem}
\label{aconv}

Now we make use of the upper bound on the noise PDF. In this section,
we state and prove a converse formulation of Theorem~\ref{th0}. Indeed
we prove that whenever $\displaystyle \cost{x} = o\left( \ln \left[
    \frac{1}{\Tu{f(x)}} \right] \right)$, the capacity-achieving input
is not bounded.
\begin{theorem}
  \label{th00}
  Whenever $\displaystyle \cost{x} = o \left( \ln \left[
      \frac{1}{\Tu{f(x)}} \right] \right)$, the support of the
  capacity-achieving input of channel~(\ref{eq:generic_ch1}) is
  unbounded.
\end{theorem}

\begin{proof}
  Suppose that the optimal input $X^*$ with distribution function
  $F^*$ has a bounded support within $[-B,B]$ for some $B > 0$. The
  KKT conditions imply that there exists $\nu \geq 0$ such that,
  \begin{equation*}
    \nu (\cost{x} - A) + C + H + \int p_N\left(y-f(x)\right) \ln p(y;F^{*})\,dy \geq 0,
  \end{equation*}
  for all $x$ in $\Reals$, with equality if $x$ is a point of increase
  of $F^*$. Using Lemma~\ref{le:iUB}, the integral term $\displaystyle
  i(f(x);F^*) = \Omega \left( \ln \left[ \frac{1}{\Tu{f(x)}} \right]
  \right)$ and hence, equation~(\ref{eq:KKT_1}) necessarily implies
  that,
  \begin{equation*}
    \nu (\cost{x} - A) + C + H  = \Omega \left( \ln \left[ \frac{1}{\Tu{f(x)}}  \right] \right),
  \end{equation*}
  which is impossible whenever $\displaystyle \cost{x} = o \left( \ln
    \left[ \frac{1}{\Tu{f(x)}} \right] \right)$.
\end{proof}


\subsection{Discreteness}

In what follows, we further characterize the capacity-achieving input
statistics when the cost function, the noise PDF and the channel
distortion function have an additional analyticity property. This
property guarantees the type of the optimal bounded input to be a
discrete one, and hence with a finite number of mass points by virtue
of compactness. This characterization permits to proceed to numerical
computations in order to compute channel capacity and find the
achieving input.

In this section, let $\eta >0$ denote a positive scalar and let
$\mathcal{S}_{\eta} = \{z \in \Complex : |\Im(z)| < \eta\}$ be a
horizontal strip in the complex domain. We adopt in this section an
alternative definition of $\Tu{x}$:
\begin{equation}
  \label{strength}
  \Tu{x} = \left\{ \begin{array}{ll}
      \displaystyle \sup_{\zeta \in  \mathcal{S}_{\eta}: \Re(\zeta) \geq x}\left|p_N(\zeta)\right| & \quad x \geq 0\\ 
      \displaystyle  \sup_{\zeta \in  \mathcal{S}_{\eta}: \Re(\zeta) \leq x}\left|p_N(\zeta)\right| & \quad x < 0,
    \end{array} \right.
\end{equation} 
and we assume that the following condition holds: The integral
$\displaystyle - \int_{-\infty}^{+\infty} \Tu{x} \ln \Tl{x}\,dx$
exists and is finite. Note that this condition is similar to C8 but
it is function of a redefined $\Tu{}$. One may think of the condition
as more restrictive. However, this strengthened condition is needed
only to establish discreteness. In the remainder of this document we
will refer to this condition as ``the strengthened-C8''.
We present hereafter, a lemma that guarantees the analyticity of
$i(\cdot;F)$ on $\mathcal{S}_{\eta} $:
\begin{lemma}
  \label{lemanalytpro}
  Whenever there exists an $\eta > 0$ such that $p_N(\cdot)$ admits an
  analytic extension on $\mathcal{S}_{\eta} $, the function
  $i(\cdot;F) : \mathcal{S}_{\eta} \rightarrow \Complex$ defined by:
  \begin{equation}
    z \rightarrow i(z;F) = - \int_{-\infty}^{\infty} \, p_N(y-z)\,\ln\,p(y;F)\,dy,
    \label{eq:def_hstable}
  \end{equation}
  is analytic.
  \label{analintstable}
\end{lemma}

\begin{proof}
  To prove this lemma, we will make use of Morera's theorem:

  {\bf a)} We start first by proving the {\em continuity\/} of
  $i(\cdot;F)$. In fact, let $\rho>0$, $z_0$ and $z \in
  \mathcal{S}_{\eta}$ such that $|z-z_0| \leq \rho$,
  \begin{align}
    \lim_{z \rightarrow z_0} i(z;F) = & - \lim_{z \rightarrow z_0}\int\,p_N(y-z)\,\ln\,p(y;F)\,dy \nonumber\\
    = & - \int \lim_{z \rightarrow z_0} \,p_N(y-z)\,\ln p(y;F)\,dy \label{stableneweq4}\\
    = & - \hspace{-3pt} \int p_N(y-z_0) \ln p(y;F) \, dy = i(z_0;F).\label{stableneweq5}
  \end{align}

  Equation~(\ref{stableneweq5}) is justified by $p_N(y-z)$ being a
  continuous function of $z$ on $\mathcal{S}_{\eta}$ by virtue of its
  analyticity and equation~(\ref{stableneweq4}) by Lebesgue's
  Dominated Convergence Theorem (DCT). Indeed, in what follows we find
  an integrable function $r(y)$ such that,
  \begin{equation*}
    \big| p_N(y-z) \, \ln\,p(y;F) \big| = - \big|p_N(y-z)\big| \ln\,p(y;F) \leq r(y),
  \end{equation*}
  for all $z \in \mathcal{S}_{\eta}$ such that $|z-z_0| \leq \rho$ and
  for all $y \in \Reals$. We upper bound first $|p_N(y-z)|$: let $y_0$
  be large enough so that Lemma~\ref{lemlow} holds
  \begin{itemize}
  \item[$\bullet$] If $y \leq -(y_0+|\Re(z_0)| + \rho)$, then $y \leq
    -y_0+\Re(z_0) - \rho$ (where $y_0$ has been defined in
    Lemma~\ref{lemlow}) and
    \begin{align*}
      \left| p_N(y-z) \right| \, \leq \, & \, \Tu{y-\Re(z)}
      \leq  \max_{\zeta \in  \mathcal{S}_{\eta}: |\zeta-z_0| \leq \rho} \Tu{y-\Re(\zeta)} 
      = \,  \Tu{y-\Re(z_0)+\rho},
    \end{align*}
    where the last equality is due to the fact that for $x \leq 0$,
    $\Tu{x}$ is non-decreasing, and for $\zeta \in \mathcal{S}_{\eta};
    |\zeta-z_0| \leq \rho$, $(y-\Re(\zeta)) \leq (y - \Re(z_0) + \rho)
    < 0$.
  \item[$\bullet$] Similarly, for $y \geq (y_0 + |\Re(z_0)| + \rho)
    \geq (y_0 + \Re(z_0) + \rho)$,
    \begin{equation*}
      \big|p_N(y-z)\big| \leq \Tu{y-\Re(z_0)-\rho}.
    \end{equation*}
  \end{itemize}  
  
  Next, using Lemma~\ref{lemlow} we also upper bound $-\ln\,p(y;F)$ to
  obtain:
  \begin{equation*}
    r(y) = \left\{ \begin{array}{ll}
        \displaystyle \Tu{y-\Re(z_0)+\rho}\,\ln \left[ \frac{2}{\Tl{y-y_0}} \right] 
        & y \leq -( y_0 + |\Re(z_0)| + \rho)\\ 
        \displaystyle -M\ln M^{'} &  |y| <  y_0 + |\Re(z_0)| + \rho\\
        \displaystyle  \Tu{y-\Re(z_0)-\rho} \, \ln \left[ \frac{2}{\Tl{y+y_0}} \right]
        & y \geq  y_0 + |\Re(z_0)| + \rho,
      \end{array} \right.
  \end{equation*} 
  where
  \begin{equation*}
    M = \max_{ \left\{ |y| \leq (y_0 + |\Re(z_0)| + \rho) \right\} }\, \max_{\{\zeta \in
      \mathcal{S}_{\eta}: |\zeta -z_0| \leq \rho\}} |p_N(y-\zeta)|
    \quad \& \quad  M^{'}  = \min_{ \left\{ |y| \leq (y_0 + |\Re(z_0)| + \rho) \right\} } \, p_Y(y;F).
  \end{equation*}
  Note that $M$ is finite since $p_N(\cdot)$ is analytic and the
  maximization is taken over a compact set, and $0 < M^{'} < 1$, since
  $p_Y(\cdot;F)$ is positive, continuous and less than $1$. Properties
  C7 and strengthened-C8 insure the integrability of $r(y)$ which
  concludes the proof of continuity of $i(z;F)$.

  {\bf b)} To continue the proof of analyticity, we need to integrate
  $i(\cdot;F)$ on the boundary $\partial \Delta$ of a compact triangle
  $\Delta \subset \mathcal{S}_{\eta}$.  We denote by $|\Delta|$ its
  perimeter, $\eta_0 = \min_{z \in \partial \Delta} \Re(z)$, $\eta_1 =
  \max_{z \in \partial \Delta} \Re(z)$ and $\phi = y_0 + \max\{
  |\eta_0|, |\eta_1| \}$. By similar arguments as above, we have
  \begin{multline*}
    \int_{\Reals} \int_{\partial \Delta} \, |p_N(y-z)| |\ln p(y;F)| dz \, dy\\
    \leq \; |\Delta| \, M^{''} \hspace{-0.3cm} \int\limits_{|y| \leq \phi} \hspace{-0.2cm} \bigl| \ln p(y;F) \bigr| dy 
    + |\Delta| \hspace{-0.3cm} \int\limits_{y \leq - \phi} \hspace{-0.2cm} \Tu{y-\eta_0} \ln 
    \left[ \frac{2}{\Tl{y-y_0}}\right] dy + |\Delta| \hspace{-0.2cm} \int\limits_{y \geq \phi} \hspace{-0.2cm} 
    \Tu{y-\eta_1} \ln \left[\frac{2}{\Tl{y-y_0}}\right] dy < \infty,
  \end{multline*}
  where
  \begin{equation*}
    M^{''} = \max_{y : |y| \leq \phi}\, \max_{\xi \in \partial \Delta}\,|p_N(y-\xi)| < \infty.
  \end{equation*}
  
  Using Fubini's theorem to interchange the order of integration,
  \begin{align}
    \int_{\partial \Delta} \hspace{-0.2cm}i(z;F)dz & = - \int_{\partial \Delta} \int_{\Reals} p_N(y-z) 
    \ln p(y;F)\,dy\,dz 
    = - \int_{\Reals}\int_{\partial \Delta} \hspace{-0.2cm}p_N(y-z)\ln p(y;F)\,dz\,dy \nonumber \\
    &= - \int_{\Reals} \ln p(y;F)\,\int_{\partial \Delta} \hspace{-0.2cm}p_N(y-z)\,dz\,dy = 0, \label{moreranoise1}
  \end{align}
  where~(\ref{moreranoise1}) is justified by the fact that $p_N(y-z)$
  is analytic for all $z \in \mathcal{S}_{\eta}$ and $y \in
  \Reals$. Equation~(\ref{moreranoise1}) in addition to the continuity
  of $i(\cdot;F)$ insure its analyticity on $\mathcal{S}_{\eta}$.
\end{proof}  
 
\begin{theorem}
  \label{fmass}
  Assume there exists an $\eta > 0$ such that $p_N(x)$ is analytically
  extendable on $\mathcal{S}_{\eta}$, and let $\mathcal{I}$ be an
  unbounded closed interval of $\Reals$\footnote{We consider that
    $\Reals$ is both closed and open.}. The capacity-achieving input
  of channel~(\ref{eq:generic_ch1}) is compactly supported and
  discrete with finite number of mass points on $\mathcal{I}$,
  whenever the following conditions hold:
  \begin{itemize}
  \item $\cost{x} = \omega \left( \ln \left[ \frac{1}{\Tl{f(x)}}
      \right] \right)$.
  \item The restrictions of $f(x)$ and $\cost{x}$ on $\mathcal{I}$
    admit analytic extensions to $\mathcal{I} \times \Reals$, denoted
    $f_{\mathcal{I}}(\cdot)$ and $\mathcal{C}_{\mathcal{I}}(\cdot)$
    respectively.
  \item The inverse map $f^{-1}_{\mathcal{I}} (\cdot)$ of
    $f_{\mathcal{I}}(\cdot)$ conserves connectedness.
  \end{itemize} 
  \label{specases}
\end{theorem}

Before we prove the theorem, we note that a necessary condition for
analytical extendability is to have $\cost{x}$ an explicit function of
the variable $x$ on $\mathcal{I}$ which can be possibly realized when
$\mathcal{I}$ is for example a subset of either $\Reals^{+}$ or
$\Reals^{-}$.

\begin{proof}
  We start by setting some notation and making a few remarks:
  \begin{itemize}
  \item Define $\mathcal{J}$ to be the image of interval $\mathcal{I}$
    by $f_{\mathcal{I}}(\cdot)$.

    Since by analyticity $f_{\mathcal{I}}(\cdot)$ is continuous, then
    $\mathcal{J}$ is an interval of $\Reals$ because $f_{\mathcal{I}}
    (\cdot)$ is identical to $f(\cdot)$ on $\mathcal{I}$, and is real
    valued.

  \item Let $\mathcal{J}_{\eta} = \{z \in \mathcal{S}_{\eta} : \Re(z)
    \in \mathcal{J}\}$ and define $\mathcal{I}_{\eta}
    =f^{-1}_{\mathcal{I}}(\mathcal{J}_{\eta})$, the inverse image of
    $\mathcal{J}_{\eta}$ by $f_{\mathcal{I}}(\cdot)$.
  
    Note that since $\mathcal{J}$ is an interval, $\mathcal{J}_{\eta}$
    is connected and so is $\mathcal{I}_{\eta}$ by virtue of the
    properties of $f^{-1}_{\mathcal{I}}(\cdot)$. Additionally, since
    $f_{\mathcal{I}}(\mathcal{I}) = \mathcal{J}$ then $\mathcal{I}
    \subset \mathcal{I}_{\eta}$.

    In what follows, we work using the induced topology on
    $\mathcal{I}_{\eta}$. Under this topology, $\mathcal{I}_{\eta}$ is
    both open and closed.
  \end{itemize}

  We proceed with the proof and assume that the optimal input $X^*$
  with distribution function $F^*$ has at least one point of increase
  in $\mathcal{I}$ for otherwise the result becomes trivial. {\em
    Assume that the points of increase of $F^*$ in $\mathcal{I}$ are
    accumulating}, and let
  \begin{equation*}
    s(z) = \nu \, (\mathcal{C}_{\mathcal{I}}(z) - A) + C + H - i(f_{\mathcal{I}}(z);F^*).
  \end{equation*}

  By the result of Lemma~\ref{lemanalytpro}, $i(f_{\mathcal{I}}(z);
  F^*)$ is analytic on $\mathcal{I}_{\eta}$ since it is the
  composition of two analytic functions: $f_{\mathcal{I}}(\cdot)$ on
  $\mathcal{I}_{\eta}$ and $i(\cdot;F^*)$ on $\mathcal{J}_{\eta} =
  f_{\mathcal{I}}(\mathcal{I}_{\eta}) \subset \mathcal{S}_{\eta}$.
  This implies that the function $s(z)$ is analytic on
  $\mathcal{I}_{\eta}$. Since by assumption the points of increase of
  $F^{*}$ have an accumulation point on $\mathcal{I}$ then by the KKT
  conditions, $s(z)$ has accumulating zeros on $\mathcal{I} \subset
  \mathcal{I}_{\eta}$, which necessarily implies by the identity
  Theorem~\cite[sec. 66]{church} that $s(\cdot)$ is identically null
  on $\mathcal{I}_{\eta}$, since $\mathcal{I}_{\eta}$ is open and
  connected. Therefore,
  \begin{equation*}
    \nu (\cost{x} - A) + C + H = - \int p_N(y) \ln p(y-f(x);F^{*})\,dy, \qquad \forall x \in \mathcal{I}.
  \end{equation*}
  
  Since $\mathcal{I}$ is unbounded, this equality is impossible for
  large values of $x$ by the result of Theorem~\ref{th0} unless $\nu =
  0$ which is non sensible. This leads to a contradiction and rules
  out the assumption of having an accumulation point on $\mathcal{I}$.
  Since $\Reals$ is Lindelof, $X^*$ is necessarily discrete on
  $\mathcal{I}$. Additionally, since the support of $X^*$ is compact
  and $\mathcal{I}$ is closed in $\Reals$, $X^*$ has necessarily a
  finite number of mass points on $\mathcal{I}$.
\end{proof}

\begin{theorem}
  \label{genas}
  Assume there exists an $\eta > 0$, such that $p_N(x)$ is
  analytically extendable on $\mathcal{S}_{\eta}$, and let
  $\mathcal{I}$ be an unbounded closed interval of $\Reals$. Whenever
  the input is constrained to have a compact support $\set{X}$, the
  capacity-achieving input is discrete with a finite number of mass
  points on $\set{X} \cap \mathcal{I}$ if the following holds:
  \begin{itemize}
  \item The restriction of $f(x)$ on $\set{X} \cap \mathcal{I}$
    admits an analytic extension to $\mathcal{I} \times \Reals$,
    denoted $f_{\mathcal{I}}(\cdot)$.
  \item The inverse map $f^{-1}_{\mathcal{I}} (\cdot)$ of
    $f_{\mathcal{I}} (\cdot)$ conserves connectedness.
  \end{itemize} 
  \label{nfmass}
\end{theorem}

Before proving the theorem, we note that the condition that the
support of $\set{X}$ is compact is a generalization of the peak power
constraint. Also it makes sense to consider sets $\set{X}$ that are
not discrete, for otherwise the problem is ill defined.

\begin{proof}
  We first note that the KKT conditions are valid under the setup of
  the compactly-supported input constraint: Indeed, the input space is
  compact in the weak topology and convex
  (see~\cite{SMITH71,fahsj}). Also note that there exists a cost
  function $\cost{x}$ the tail of which is $\omega \left( \ln
    \left|f(x) \right| \right)$ and such that $\sup_{\set{X}}\cost{x}
  = A$, for some $A > 0$.

  Now, for any $F \in \mathcal{P}_{\mathcal{X}}$ --the set of all
  input distributions having a compact support $\set{X}$, we have
  $\int \cost{x}\,dF(x) \leq A$ which implies that
  $\mathcal{P}_{\mathcal{X}} \subset \mathcal{P}_A$. Since the mutual
  information is finite, continuous and weakly differentiable on
  $\mathcal{P}_A$ whenever $\cost{x} = \omega \left( \ln \left| f(x)
    \right| \right)$ (see Appendix~\ref{ap:Prop_I}) then it is as such
  on $\mathcal{P}_{\mathcal{X}}$. Under this setup, the KKT conditions
  state that an input RV $X^*$ with CDF $F^*$ achieves the capacity
  $C$ of a compact-support constrained channel if and only if,
  \begin{equation*}
    C + H + \int p_N\left(y-f(x)\right) \ln p(y;F^{*})\,dy \geq 0, \qquad \forall x \in \set{X},
  \end{equation*}
  with equality if $x$ is a point of increase of $F^*$, and where $H$
  is the entropy of the noise. By virtue of the analyticity
  conditions, the function $s(z) = C + H + \int
  p_N\left(y-f_{\mathcal{I}}(z)\right) \ln p(y;F^{*})\,dy$ would also
  be analytic on $\set{I}$. The assumption that the points of increase
  of $X^*$ on $\set{X} \cap \set{I}$ have an accumulation point is
  impossible since it will lead by the identity theorem to $s(x) = 0$
  on $\set{I}$ which is impossible since $i(x;F^*) = - \int
  p_N\left(y-f(x)\right) \ln p(y;F^{*}) = \Omega \left( \ln \left[
      \frac{1}{\Tu{f(x)}} \right] \right)$ (see the proof of
  Theorem~\ref{th00}), which increases to $\infty$. Therefore, $X^*$
  is necessarily discrete on $\set{X} \cap \set{I}$. The finiteness of
  the number of mass points is a direct consequence of the compactness
  of $\set{X} \cap \set{I}$.
\end{proof}

\begin{note*}
  A similar statement to that of Theorem~\ref{nfmass} may be made
  whenever, in addition to a compact support constraint, there is also
  a cost constraint satisfying the conditions of Theorem~\ref{fmass}
  with tail behavior either $\omega \left( \ln \left[
      \frac{1}{\Tl{f(x)}} \right] \right)$ or $o \left( \ln \left[
      \frac{1}{\Tu{f(x)}} \right] \right)$.
\end{note*}

Before moving to giving some concrete examples to our general
theorems, we would like to state that some conditions were only
considered for either the sake of the clarity of the proofs, or for
conserving the general aspect of the results. Many such conditions
could be relaxed while conserving some or all of the found
conclusions. For example,
\begin{itemize}
\item The notions of $\omega$, $\Omega$, $o$ and $O$ used in this
  document are defined as $|x| \rightarrow +\infty$, i.e., in such a
  way to capture a symmetric rate of decay for both tails. However,
  one can only consider left or right tail behaviors separately. The
  results of boundedness and discreteness could be given in terms of
  each tail where for example for non-symmetric noise PDFs or
  non-symmetric cost functions, the optimal input could only be
  bounded on one of the semi-axis.
  
\item For Theorems~\ref{th0} and~\ref{th00}, the assumption that
  $p_N(\cdot)$ is positive could be relaxed to one sided noise
  PDFs. These theorems are still valid on one side of the axis.


\item The proven theorems --stated in terms of $\Tl{x}$ and $\Tu{x}$--
  could be stated in terms of any two functions having the same
  properties and providing lower and upper bounds on $p_N(x)$ for
  large values of $|x|$.

\end{itemize}


\section{Applications of the Theorems and Numerical Results}
\label{exnum}

In this section we apply our results to a variety of specific channels
of interest that fit under the general framework presented
previously. For those channels that have been previously studied in
the literature, we verify our results --in the form of
Theorems~\ref{th0},~\ref{th00},~\ref{fmass} and~\ref{nfmass}, and for
the other models we state some new results. We note that in all the
examples presented subsequently the considered functions $f(\cdot)$
and the cost constraints satisfy the general conditions C1 through C4
in Section~\ref{conditions}. The noise distributions are absolutely
continuous with positive, continuous PDFs with tails that have ``at
least'' a polynomial decay and hence satisfying the assumptions C5 and
C6. Finally, in all the provided examples the noise PDFs possess a
monotonic tail and a finite differential entropy and therefore,
condition C8 is satisfied. It remains to check for each example
condition C7 and possibly the strengthened-C8. 

For the purpose of verifying condition C7, we note that one can use
$p_N (x)$ instead of $\Tl{x}$ since they are identical at large values
of $|x|$. When it comes to discreteness, whenever $|x|$ is large
enough the function $\Tu{x}$ defined in~(\ref{strength}) becomes
\begin{equation*}
  \Tu{x} = \sup_{ \left\{z: \quad \Re(z) = x \, \& \, |\Im(z)| < \eta \right\}} \left| p_N(z) \right|,
\end{equation*}
because $\left| p_N(z) \right|$ is decreasing with $|\Re(z)|$ at large
values for all the given examples.

For each model we consider in what follows, we will check whether the
appropriate conditions are satisfied, state the results --specialized
to the channel at hand, and compare with the known results in the
literature.

\subsection{The Gaussian Model}

For a Gaussian noise distribution with mean zero and variance
$\sigma^2$, the PDF is $p_N(x) = \frac{1}{\sqrt{2 \pi \sigma^2}}
e^{-\frac{x^2}{2\sigma^2}}$ and we write $N \sim
\Normal{0}{\sigma^2}$.

\underline{Checking the conditions:} Condition C7 is validated as
follows: for large values of $|x|$ and $|y|$,
\begin{eqnarray*}
  L(x+y) & = & \ln \left[ \frac{1}{p_N(x + y)} \right] 
    = \ln \sqrt{2 \pi \sigma^2} +  \frac{(x+y)^2}{2\sigma^2} \\
  & \leq & 2 \left(\ln \sqrt{2 \pi \sigma^2}  + \frac{x^2}{2\sigma^2} 
    + \ln \sqrt{2 \pi \sigma^2} + \frac{y^2}{2\sigma^2}\right) - 3 \ln \sqrt{2 \pi \sigma^2} 
  = \kappa_{\text{l}} \left( L(x) + L(y) \right),
\end{eqnarray*}
where $\kappa_{\text{l}} > 2$. When it comes to discreteness, let
$p_N(z) = \frac{1}{\sqrt{2 \pi \sigma^2}} e^{-\frac{z^2}{2\sigma^2}}$,
be an analytic extension of $p_N(x)$ to the complex plane, where $z =
x + j y$. The magnitude of $p_N(z)$ is
\begin{equation*}
  \left|p_N(z)\right| = \frac{1}{\sqrt{2 \pi \sigma^2}} \left|e^{-\frac{z^2}{2\sigma^2}}\right| 
  = \frac{1}{\sqrt{2 \pi \sigma^2}} e^{-\frac{x^2-y^2}{2\sigma^2}}, 
\end{equation*} 
and is decreasing in $x=\Re(z)$. Therefore, $\Tu{x} = \frac{1}{\sqrt{2
    \pi \sigma^2}} e^{-\frac{x^2-\eta^2}{2\sigma^2}} =
e^{\frac{\eta^2}{2\sigma^2}} \, p_N(x)$.

Checking for the strengthened-C8, the integral $ -
\int_{-\infty}^{+\infty} \Tu{x} \ln \Tl{x}\,dx = e^{\frac{\eta^2}{2
    \sigma^2}} h(N)$ which is finite because the noise differential
entropy $h(N)$ is finite.

The following theorem is a specialization of Theorems~\ref{th0} and
\ref{th00} for this specific Gaussian case:
\begin{theorem}
  \label{th0g}
  Whenever $\cost{x} = o\left(f(x)^2\right)$, the support of the
  capacity-achieving input of channel~(\ref{eq:generic_ch1}) when
  $N\sim \Normal{0}{\sigma^2}$ is unbounded. 

  Whenever $\cost{x} = \omega\left(f(x)^2\right)$, the support of the
  capacity-achieving input of channel~(\ref{eq:generic_ch1}) when $N
  \sim \Normal{0}{\sigma^2}$ is compact. Furthermore, the optimal
  input is discrete with finite number of mass points whenever
  $\mathcal{C}(\cdot)$ and $f(\cdot)$ satisfy the analyticity and
  connectedness conditions of Theorem~\ref{fmass}. 

  Under a compact support constraint, the optimal input is also
  discrete with finite number of mass points whenever $f(\cdot)$
  satisfies the conditions of Theorem~\ref{nfmass}.
\end{theorem}

\underline{Previous work:} A possibly non-linear ($f(x) = x^n$, $n \in
\Naturals^*$) Gaussian channel under an even moment constraint
($\cost{x} = x^{2k}$) was considered in~\cite{fahsj} as a core channel
model from which results on multiple non-linear channel models were
derived. The authors applied a standard Hilbert space decomposition
using Hermite polynomials as bases and proved that, for $n < 2k$, the
capacity-achieving distribution has the following behavior:
\begin{itemize}
\item Whenever $n = k$ 
  \begin{itemize}
  \item if $n$ is odd, the optimal input $F^*$ is absolutely continuous.
  \item if $n$ is even, $F^*$ is discrete with no accumulation points.
  \end{itemize}
\item Whenever $n < k$, $F^*$ is discrete with finite number of mass
  points.
\item Whenever $k < n < 2k$, $F^*$ is discrete with no accumulation
  points.
\end{itemize}

We point out that while the results stated in Theorem~\ref{th0g} do
not cover the limiting case $n = k$ --which corresponds to the case
$\cost{x} = \theta\left(f^2(x)\right)$, the result for the case ``$n <
k$'' is identical. Whenever $k < n$, Theorem~\ref{th0g} states that
the support of $F^*$ is not bounded; a conclusion that could not be
reached in~\cite{fahsj}.

\subsection{Gaussian Mixtures}
\label{gaumix}

Gaussian mixtures are widely used as more tractable models to some
non-Gaussian noise statistics~\cite{blu,amir}.  One approach in
dealing with such distributions is based on the observation that in
limiting cases Gaussian mixtures are nearly Gaussian and they are
simplified accordingly. The PDF of a Gaussian mixture RV is:
\begin{equation*}
  p_N(x) = \sum_{i=1}^{n}\alpha_i \, p_{N_{i}}(x),
\end{equation*} 
where $n \in \Naturals^*$ and for $1 \leq i \leq n$,
\begin{itemize}
\item[$\bullet$] $N_i \sim \Normal{\mu_i}{\sigma^2_i}$ are Gaussian
  RVs with mean $\mu_i$ and variances $\sigma^2_i \neq 0$. We assume
  WLOG that $\sigma_1 \geq \cdots \geq \sigma_n$.
\item[$\bullet$] $0 \leq \alpha_{i} \leq 1$, and
  $\sum_{i=1}^{n}\alpha_i = 1$.
\end{itemize}

Before proceeding, we note that the rate of decay of this noise PDF is
\begin{eqnarray*}
  \ln \left[\frac{1}{p_N(x)}\right]  &=& \ln \left[ \frac{\frac{\sqrt{2 \pi \sigma^2_1}}{\alpha_{1}}e^{\frac{\left(x-\mu_1\right)^2}{2 \sigma^2_1}}}
    {1+\sum_{i=2}^{n}\frac{\alpha_i \sigma_1}{\alpha_1\sigma_i}e^{-\frac{\left(x-\mu_i\right)^2}{2 \sigma^2_i}+\frac{\left(x-\mu_1\right)^2}{2 \sigma^2_1}}} \right]\\
  &=& \ln \left[ \frac{\sqrt{2 \pi \sigma^2_1}}{\alpha_1} \right] + \frac{\left(x-\mu_1\right)^2}{2 \sigma^2_1} - \ln \left[1+\sum_{i=2}^{n}
    \frac{\alpha_i \sigma_1}{\alpha_1\sigma_i}e^{-x^2\left(\frac{1}{2 \sigma^2_i}-\frac{1}{2 \sigma^2_1}\right)+x \left(\frac{\mu_i}{\sigma_i^2}-\frac{\mu_1}{\sigma_1^2}\right) - \left(\frac{\mu_i^2}{2\sigma_i^2}-\frac{\mu_1^2}{2\sigma_1^2}\right)} \right]\\
  &=& \ln \left[ \frac{\sqrt{2 \pi \sigma^2_1}}{\alpha_1} \right] + \frac{\left(x-\mu_1\right)^2}{2 \sigma^2_1} - \Theta\left(\sum_{i=2}^{n}
    \frac{\alpha_i \sigma_1}{\alpha_1\sigma_i}e^{-x^2\left(\frac{1}{2 \sigma^2_i}-\frac{1}{2 \sigma^2_1}\right)+x \left(\frac{\mu_i}{\sigma_i^2}-\frac{\mu_1}{\sigma_1^2}\right) - \left(\frac{\mu_i^2}{2\sigma_i^2}-\frac{\mu_1^2}{2\sigma_1^2}\right)} \right)\\
  &=& \Theta(x^2).
\end{eqnarray*}

\underline{Checking the conditions:} Since in
Section~\ref{conditions}, we proved that condition C7 defines a
convex set of functions, then by the results of the Gaussian model,
each $p_{N_{i}}(\cdot)$ satisfies condition C7 and so does
$p_{N}(\cdot)$. To study discreteness, we let $p_N(z) =
\sum_{i=1}^{n}\alpha_i \, p_{N_{i}}(z)$, be an analytic extension of
$p_N(x)$ on the complex plane. Since
\begin{equation*}
  |p_N(z)| \leq \sum_{i=1}^{n}\alpha_i\,|p_{N_{i}}(z)| = \sum_{i=1}^{n}\frac{\alpha_i}{\sqrt{2 \pi \sigma^2_i}}e^{-\frac{(x-\mu_i)^2-y^2}{2\sigma^2_i}},
\end{equation*}
then, for a large-enough $|x|$,
\begin{equation*}
  \Tu{x} = \sup_{ \left\{z: \quad \Re(z) = x \, \& \, |\Im(z)| < \eta \right\}} \left| p_N(z) \right|
  \, \, \leq \, \, e^{\frac{\eta^2}{\sigma^2_n}} \, \, \sum_{i=1}^{n}\frac{\alpha_i}{\sqrt{2 \pi \sigma^2_i}}\, 
  e^{-\frac{(x-\mu_i)^2}{2\sigma^2_i}} = e^{\frac{\eta^2}{\sigma^2_n}} \, p_N(x),
\end{equation*}
which implies that strengthened C8 is valid of the finiteness of
\begin{equation*}
  - \int_{-\infty}^{+\infty} \Tu{x} \ln \Tl{x}\,dx,
\end{equation*}
as $h(N)$ is finite by virtue of the fact that $N$ has a finite
variance $\sigma^2 = \sum_{i=1}^{n}\alpha_i \sigma^2_i$.

Specializing the results to the channel at hand, we can state the
following:
\begin{theorem}
  \label{th0mg}
  Whenever $\cost{x} = o\left(f(x)^2\right)$, the support of the
  capacity-achieving input of channel~(\ref{eq:generic_ch1}) when $N$
  is a Gaussian mixture is unbounded.

  Whenever $\cost{x} = \omega\left(f(x)^2\right)$, the support of the
  capacity-achieving input of channel~(\ref{eq:generic_ch1}) when $N$
  is a Gaussian mixture is compact. Furthermore, the optimal input is
  discrete with finite number of mass points whenever
  $\mathcal{C}(\cdot)$ and $f(\cdot)$ satisfy the analyticity and
  connectedness conditions of Theorem~\ref{fmass}. 

  Under a compact support constraint, the optimal input is also
  discrete with finite number of mass points whenever $f(\cdot)$
  satisfies the conditions of Theorem~\ref{nfmass}.
\end{theorem}


\underline{Previous work:} To our knowledge, a formal analysis of
Gaussian mixtures channels has not been conducted before, and hence
Theorem~\ref{th0mg} states a new previously unknown result.  We note
that since the transitional rate of decay is $\theta
\left(f^2(x)\right)$, the capacity of the linear ($f(x) = x$) Gaussian
mixtures channel under an average power constraint ($\cost{x} = x^2$)
is not in the scope of this work. However, in~\cite{Fahs2} it was
shown that, except for Gaussian noise, the capacity of the linear
average power constrained channel is achieved by discrete statistics
for all noise distributions satisfying certain conditions, ones that
are indeed satisfied by Gaussian mixtures.

In Figure~{1}, we plot the numerically-computed capacity of a sample
Gaussian mixture channel. The results of~\cite{Fahs2} were used and an
optimal discrete input distribution that satisfies the necessary and
sufficient KKT condition was sought. The numerical computations were
conducted using Matlab. 


\begin{figure}[htb]
  \begin{center}
    \includegraphics[width=4in]{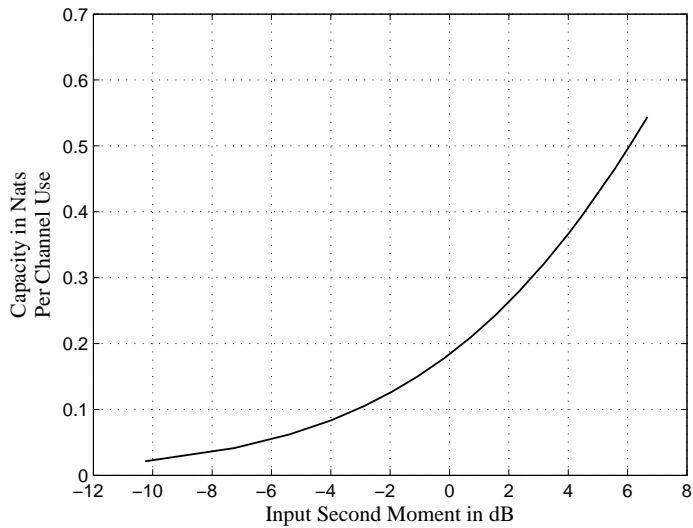}
    \caption{\small Capacity of the linear channel under the Gaussian mixture noise 
      $p_N(x) = 0.5 \,p_{N_1}(x) + 0.5\, p_{N_2}(x)$ where $N_1 \sim \mathcal{N}(0,1)$ 
      and $N_2 \sim \mathcal{N}(0,4)$. \label{1}}
  \end{center}
\end{figure}

\subsection{Generalized Gaussian}
\label{gengau}

Generalized Gaussians~\cite{mil} are viewed as a class of
distributions generalizing the well-known Laplacian and Gaussian
distributions.
Additive noise is often assumed to be a generalized Gaussian RV in
order to model the impulsive\footnote{By impulsive it is meant that
  extreme values of the noise signal are observed very frequently
  (i.e., with notable amount of probability).} nature of noise in
communication channels~\cite{shinde,bouvet,blackard,thgon}. In other
instances, these models were considered for the ultra-wide band
multiple access interference plus noise~\cite{fiorina2006,BSF08}.

Generalized Gaussians have exponentially decaying PDFs given by:
\begin{equation}
  p_N(x)  = \frac{a}{2b\Gamma\left(\frac{1}{a}\right)}e^{-\left(\frac{|x-\mu|}{b}\right)^{a}},
  \label{eq:GenGauss}
\end{equation}
where $\Gamma(\cdot)$ is the Gamma function, $a \in \Reals^{+*}$ is a
shape parameter, $b \in \Reals^{+*}$ is a scale parameter and $\mu \in
\Reals$ is a location parameter. In the remainder of this section, we
will assume WLOG that the location parameter $\mu$ is equal to zero.

\underline{Checking the conditions:} Condition C7 is
satisfied. Indeed,
\begin{eqnarray*}
  L(x+y) & = & \ln \left[ \frac{1}{p_N(x+y)} \right] = \ln \left[ \frac{2 b \Gamma\left(\frac{1}{a}\right)}{a} \right]
  +  \frac{|x+y|^a}{b} \\
  & \leq & \ln \left[ \frac{2 b \Gamma\left(\frac{1}{a}\right)}{a} \right] + \max\{1;2^{a-1}\} \, \frac{|x|^a+|y|^a}{b} \\
  & = & \max\{1;2^{a-1}\} \, \left[ L(x)+L(y) \right] + \min\{0;1-2^{a-1}\} \ln \left[\frac{2 b \Gamma\left(\frac{1}{a}\right)}{a} \right]\\
  & \leq & \kappa_{\text{l}}\left[ L(x)+L(y) \right],
\end{eqnarray*}
for some $\kappa_{\text{l}} > \max\{1;2^{a-1}\}$ for large-enough
values of $|x|$ and $|y|$.

One can therefore state the following theorem:
\begin{theorem}
  \label{th0gen}
  Whenever $\cost{x} = o\left(\left|f(x)\right|^a\right)$, the support
  of the capacity-achieving input of channel~(\ref{eq:generic_ch1})
  when $N$ is a generalized Gaussian RV~(\ref{eq:GenGauss}) is
  unbounded.

  Whenever $\cost{x} = \omega\left(\left|f(x)\right|^a\right)$, the
  support of the capacity-achieving input of
  channel~(\ref{eq:generic_ch1}) when $N$ is a generalized Gaussian
  RV~(\ref{eq:GenGauss}) is compact.
\end{theorem}


\underline{Previous work:} To our knowledge, no previous information
theoretic work has appeared regarding this channel model. For the
linear channel under an average power constraint for instance, the
optimal input of channel~(\ref{eq:generic_ch1}) is bounded whenever
the noise is a generalized Gaussian with parameter $a < 2$.

\subsection{Polynomially-Tailed Distributions} 

Gaussian mixtures and generalized Gaussians are considered by many
researchers to fail to capture the ``impulsiveness'' of the
noise. This failure is due to several reasons, the most important of
which is that they do not possess the algebraic behavior of
heavy-tailed noise distributions encountered in typical communication
channels~\cite{Shaop}. One family of such distributions, the
``generalized Cauchy''~\cite{mil}, is found to be reasonable in modeling
the amplitude of atmospheric impulse noise~\cite{kassam}. In this
document, we refer by ``polynomially-tailed'' noise distributions to
all distributions satisfying
\begin{equation*}
  p_N(x) = \Theta\left(\frac{1}{|x|^{1+\alpha}}\right), \quad \text{for some } \alpha > 0,
\end{equation*}
which include among others: the Gamma, Pareto (one sided) and
alpha-stable distributions.

\underline{Checking the conditions:} In order to proceed, we use the
``obvious'' lower and upper bounds on $p_N(x)$ for large values of
$|x|$ instead of $p_N(x)$ itself and we state the corresponding
theorems accordingly. These bounds are of the form
$\frac{\zeta_{\text{l}}}{|x|^{1+\alpha}}$ and
$\frac{\zeta_{\text{u}}}{|x|^{1+\alpha}}$, for some $\zeta_{\text{l}}$
and $\zeta_{\text{u}} >0$. We prove now that condition C7 is
satisfied; Let
\begin{equation*}
  L(x) = \ln \left[ \frac{|x|^{1+\alpha}}{\zeta_{\text{l}}} \right] = (1+\alpha) \ln |x| -\ln \zeta_{\text{l}},
\end{equation*}
which implies that for large-enough $|x|$ and $|y|$,
\begin{eqnarray}
  L(x+y) = (1+\alpha)\ln|x+y| -\ln \zeta_{\text{l}} & \leq & (1+\alpha) \ln \left[ |x|+|y| \right] -\ln \zeta_{\text{l}} \nonumber\\
  & \leq & p(1+\alpha) \left[ \ln |x| + \ln |y| \right] - \ln \zeta_{\text{l}} \nonumber\\
  & = & p \left[ (1+\alpha) \ln |x| - \ln \zeta_{\text{l}} + (1+\alpha) \ln |y| - \ln \zeta_{\text{l}} \right] + (2 p-1) \ln\zeta_{\text{l}}\nonumber\\
  & \leq & 2 p \left[ (1+\alpha) \ln |x| -\ln \zeta_{\text{l}} + (1+\alpha) \ln |y| - \ln \zeta_{\text{l}} \right] \nonumber\\
  & = & 2 p \left[ L(x) + L(y) \right] \nonumber,
\end{eqnarray}
where $p > 1$. Consequently, the following holds:
\begin{theorem}
  \label{th0poly}
  Whenever $\cost{x} = \omega \left( \ln |f(x)| \right)$, the support
  of the capacity-achieving input of channel~(\ref{eq:generic_ch1})
  when $N$ is polynomially-tailed is compact.
\end{theorem}

For example, for a linear channel subjected to an additive
polynomially-tailed noise, the optimal input has a bounded support for
any cost function that is super logarithmic (i.e., $\omega\left(\ln
  |x|\right)$) such as the average power constraint.

Note that the other ``range'' $\cost{x} = o \left( \ln |f(x)| \right)$
is outside the scope of this work as condition C4 will not be
satisfied. When it comes to discreteness and strengthened-C8, it
depends on the analyticity property of the specific $p_N(\cdot)$ under
consideration.

The remaining part of this Section is dedicated to two important types
of polynomially decaying distributions, for which we prove that the
discreteness results of Theorem~\ref{fmass} apply.

\subsubsection{Non-Totally Skewed Alpha-Stable and their Mixtures}
\label{asm}

The term ``stable'' is used because, under some constraints, these
distributions are closed under convolution. The stable distributions,
which are a subset of that of infinitely divisible distributions, are
the only laws that have the captivating property of being the
resultant of a limit of normalized sums of IID RVs. This result is
referred to as the Generalized Central Limit Theorem (GCLT), a
property that constitutes one of the main reasons behind the adoption
of Gaussian statistics for noise models in communication channels.

Though the Gaussian distribution is one of the stable laws, it
represents the exception: it is unique in the sense that it is the
only one that has a finite variance and an exponential tail; All
others have an infinite variance and a polynomial tail. A complete
literature on the theory of stable distributions can be found in
\cite{fel,kolmo,zolo,newzolo}. In this document we use the term
``alpha-stable'' to refer to stable variables {\em other than the
  Gaussian\/}. Although only few alpha-stable RVs have closed form
densities (namely the Cauchy and the L\'{e}vy laws), these
distributions are well characterized in the Fourier domain: The
characteristic function\footnote{The characteristic function
  $\phi(t)$ of a distribution function $F(x)$ is defined by:
  \begin{equation*}
    \phi(t) = \int_{\Reals} e^{it x}\,dF(x). 
  \end{equation*}}
of an alpha-stable RV is given by:
\begin{align*}
  \phi(t) = \exp \left[ i \delta t - \asexp{t} \right], \qquad
  & \biggl(0< \alpha < 2 \quad -1 \leq \beta \leq 1 \quad \gamma \in
  \Reals^{+*} \quad \delta \in \Reals \biggr),
\end{align*}
where $\sgn(t)$ is the sign of $t$, and the function
$\Phi(\cdot)$ is given by:
\begin{equation*}
  \Phi(t) = \left\{ \begin{array}{ll} 
      \displaystyle \tan \left(\frac{\pi \alpha}{2} \right) \quad & \alpha \neq 1 \\
      \displaystyle  -\frac{2}{\pi} \ln|t| \quad & \alpha = 1.
    \end{array} \right.
\end{equation*} 

The constant $\alpha$ is called the ``characteristic exponent'',
$\beta$ is the ``skewness'' parameter, $\gamma$ is the ``scale''
parameter ($\gamma^{\alpha}$ is often called the ``dispersion'') and
$\delta$ is the ``location'' parameter. Such a RV will be denoted $N
\sim \mathcal{S}(\alpha,\beta,\gamma,\delta)$.

In what follows, we limit our analysis to non-totally skewed
alpha-stable variables, i.e., ones for which $|\beta| \neq 1$.

\underline{Checking the conditions:} For non-totally skewed laws, both
the right and the left tails are polynomially decaying as
$\Theta\left(\frac{1}{|x|^{\alpha+1}}\right)$ (see~\cite[Th.1.12,
p.14]{nolan:2012}), and Theorem~\ref{th0poly} holds. Furthermore,
whenever $\alpha \geq 1$ the alpha-stable variables are analytically
extendable, to the whole complex plane when $\alpha > 1$ and to some
horizontal strip when $\alpha = 1$~\cite[theorem 2.3.1 p. 48 and
remark 1 p. 49]{ibrlin}. We check in what follows the
strengthened-C8. We derive in Appendix~\ref{ratedec} a novel bound on
the rate of decay of the complex extension of the alpha-stable PDF
when $\alpha \geq 1$: For small-enough $\eta > 0$, there exist $\kappa
> 0$ and $n_0 > 0$ such that
\begin{equation}
  |p_{N}(z)| \leq \frac{\kappa}{|\Re(z)|^{\alpha + 1}}, \quad  \forall \, z \in \mathcal{S}_{\eta}: |\Re(z)| \geq n_0.
  \label{extbd}
\end{equation} 

This bound insures the validity of Theorems~\ref{fmass}
and~\ref{nfmass} whenever the conditions on $\mathcal{C}(\cdot)$ and
$f(\cdot)$ are satisfied, and hence the following theorem is valid:
\begin{theorem}
  \label{stath}
  Whenever $\cost{x} = \omega \left( \ln |f(x)| \right)$, the support
  of the capacity-achieving input of channel~(\ref{eq:generic_ch1})
  when $N \sim \mathcal{S}(\alpha,\beta,\gamma,\delta)$ is a
  non-totally skewed alpha-stable variable is compact.

  Whenever $\alpha \geq 1$, the optimal input is discrete with finite
  number of mass points whenever $\mathcal{C}(\cdot)$ and $f(\cdot)$
  satisfy the analyticity and connectedness conditions of
  Theorem~\ref{fmass}. 

  Under a compact support constraint, the optimal input is also
  discrete with finite number of mass points whenever $f(\cdot)$
  satisfies the conditions of Theorem~\ref{nfmass}.
\end{theorem}

Note that by virtue of the fact that condition C7 defines a convex
set, the results presented here for one alpha-stable variable are
valid for any convex combinations of them.

\underline{Previous work:} The capacity of the additive linear channel
was considered in~\cite{Fahs12}, where the noise is modeled as
symmetric alpha-stable ($\beta =0$) for the range $\alpha \geq 1$.
Subjected to a fractional $r$-th moment constraint, $\E{ \left| X
  \right|^r } \leq a$, $a >0$ and $r >1$, the optimal input was found
to be achieved by discrete statistics.  Theorem~\ref{stath}
generalizes this result to cover the non totally-skewed alpha-stable
family and generic input cost functions that are
``super-logarithmic''.  As a direct application of
Theorem~\ref{stath}, it can be seen that the result of~\cite{Fahs12}
also holds for $|\beta| \neq 1$ (not necessarily equal to $0$) and for
the range $r \leq 1$.

We use a specialized numerical MATLAB package~\cite{fahsj} to search
for the positions of the optimal points and their respective
probabilities whenever the optimal input is discrete. In Figure~{2},
we plot the capacity of channel~(\ref{eq:generic_ch1}) whenever $f(x)
= x$, $\mathcal{C}(|x|) = x^2$ and $N \sim \mathcal{S}(\alpha,0,1,0)$
for $\alpha = 1, 1.2, 1.5$ and $1.8$. The capacity curves clearly
shows that as $\alpha$ gets bigger the capacity is higher. This is in
accordance with the fact that the lower the value of $\alpha$, the
distribution becomes heavier.

\begin{figure}[htb]
  \begin{center}
    \includegraphics[width=4in]{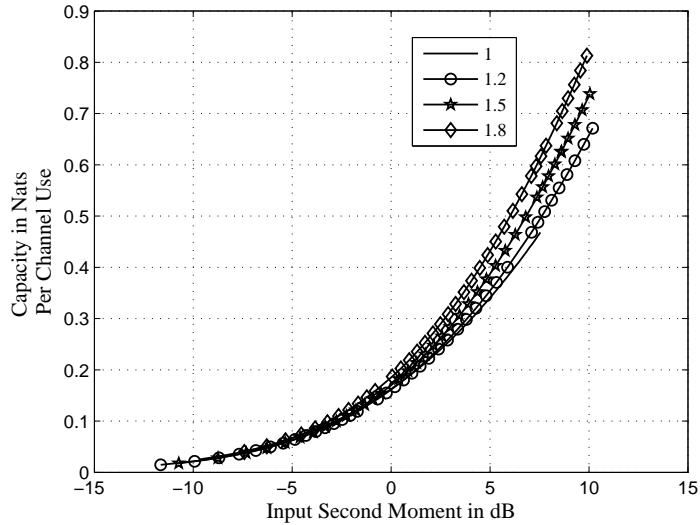}
    \caption{\small Capacity of the linear channel subject to
      symmetric ``standard'' alpha-stable noise $N \sim
      \mathcal{S}(\alpha, 0, 1, 0)$ for various values of the characteristic exponent
      $\alpha$.\label{2}}
  \end{center}
\end{figure}


\subsubsection{Composite noise: Gaussian + Alpha-Stable}

Recently, a compound noise model was adopted to capture potentially
different sources of noise: a Gaussian model for the thermal noise and
an alpha-stable model for the potential MAI, as is the case for ad-hoc
self configuring networks with applications in CDMA
networks~\cite{ElGhannoudi}, and in the general context of ultra
wideband technologies~\cite{Win}. Further information on the subject
can be found in~\cite{hughes,evans,chopra}. This noise model is widely
known as the Middleton class B model~\cite{Middle86,Kim98}.
We consider hence the following additive noise $N = N_1 + N_2$, where
\begin{itemize}
\item $N_1 \sim \mathcal{S}(\alpha,\beta,\gamma,\delta)$, which
  represents the effect of the MAI, assumed a non totally-skewed
  alpha-stable RV.
\item $N_2 \sim \Normal{\mu}{\sigma^2}$ is a Gaussian RV that models
  the effect of thermal noise.
\end{itemize} 

\underline{Checking the conditions:} It has been proved
in~\cite[Appendix I]{Fahs14} that $p_N(x)$ is polynomially-tailed
which implies that Theorem~\ref{th0poly} holds for the compound noise
model. In order to apply Theorems~\ref{fmass} and~\ref{nfmass} for the
channels impaired by the composite noise $N$, we use the fact that its
PDF is analytically extendable on $\Complex$ (for all values of $0 <
\alpha < 2$) and therefore on $\mathcal{S}_{\eta}$~\cite[Appendix
I]{Fahs14}, and check the strengthened-C8:
\begin{eqnarray*}
  \Tu{x} =  \sup_{|\Im(z)| < \eta}|p_N(z)| &\leq& \sup_{|\Im(z)| < \eta} \frac{1}{\sqrt{2 \pi \sigma^2}} \int \left| 
    e^{-\frac{(z-t)^2}{2\sigma^2}} \right|  \,p_{N_{1}}(t)\,dt \nonumber\\
  &\leq& \frac{1}{\sqrt{2 \pi \sigma^2}} e^{\frac{\eta^2}{2\sigma^2}} \int e^{-\frac{(x-t)^2}{2\sigma^2}} \,p_{N_{1}}(t)\,dt 
  = e^{\frac{\eta^2}{2\sigma^2}} p_N(x),
\end{eqnarray*}
which implies
\begin{equation*}
  - \int_{-\infty}^{+\infty} \Tu{x} \ln \Tl{x}\,dx   \leq  - e^{\frac{\eta^2}{\sigma^2}} \int_{-\infty}^{+\infty} p_N(x) 
  \ln p_N(x)\,dx = e^{\frac{\eta^2}{\sigma^2}} h(N) < \infty.
\end{equation*}

The following theorem therefore holds:
\begin{theorem}
  \label{comth}
  Whenever $\cost{x} = \omega \left( \ln |f(x)| \right)$, the support
  of the capacity-achieving input of channel~(\ref{eq:generic_ch1})
  when $N = N_1 + N_2$ is compact. The optimal input is discrete with
  finite number of mass points whenever $\mathcal{C}(\cdot)$ and
  $f(\cdot)$ satisfy the analyticity and connectedness conditions of
  Theorem~\ref{fmass}. 

  Under a compact support constraint, the optimal input is also
  discrete with finite number of mass points whenever $f(\cdot)$
  satisfies the conditions of Theorem~\ref{nfmass}.
\end{theorem}   

\underline{Previous work:} The capacity of the additive channel
subjected to the compound noise $N = N_1 + N_2$ was characterized
in~\cite{Fahs14} by the authors.

We plot in Figure~\ref{3} the capacity of the linear channel under an
input second-moment constraint whenever $N_1 \sim \mathcal{N}(0,1)$
and $N_2 \sim \mathcal{S}(\alpha,0,1,0)$ for the values of $\alpha =
1$ and $1.5$ where the optimal input at 7.27dB was found to have 16
and 18 mass points respectively.

We note that the composite noise channel cannot be approximated by a
Gaussian channel because the overall noise will be heavy tailed
whenever the stable noise is present. Indeed, the composite noise here
has infinite variance. If one where to ignore the presence of a
``mild'' stable noise component such as $N \sim
\mathcal{S}(1,0,0.1,0)$ and assume the additive noise to have only a
Gaussian component, $\frac{1}{2}\,\ln\left(1 +
  \frac{\E{X^2}}{\sigma^2}\right) = 0.356$ nats at $0.16$
dB. This is to be compared with the capacity of the composite channel
which is only $0.298$ nats/channel-use.


\begin{figure}[htb]
  \begin{center}
    \includegraphics[width=4in]{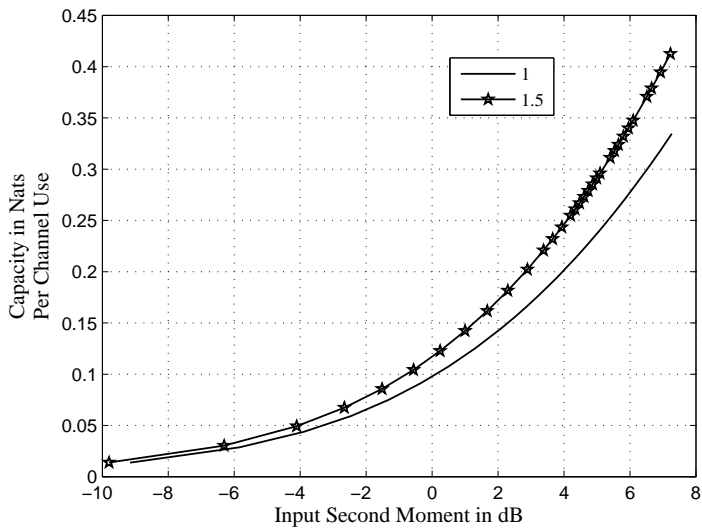}
    \caption{\small Capacity of the linear channel under the composite
      noise: a standard Gaussian \& a standard alpha-stable for
      $\alpha = 1 \, \& \, 1.5$.\label{3}}
  \end{center}
\end{figure}

\section{Conclusion}   
\label{conc} 

We studied the problem of characterizing the capacity and its
achieving distributions for additive noise channels of the form $Y =
f(X) + N$, where the input is subjected to an input cost constraint of
the form $\E{\mathcal{C} \left( |X| \right) } \leq A$, $A > 0$. We
proved that the type of the optimal input is intimately related to the
growth rate at infinity of the functions $f(x)$,
$\mathcal{C}\left(|x|\right)$ and $\frac{1}{p_N(x)}$ through a simple
relationship. Indeed, for monotonically tailed noise density functions
whenever $\mathcal{C}\left(|x|\right) =
\omega\left(\ln\left[\frac{1}{p_N\left(f(x)\right)}\right]\right)$,
the support of the optimal input is necessarily bounded. Conversely,
if $\mathcal{C}\left(|x|\right) =
o\left(\ln\left[\frac{1}{p_N\left(f(x)\right)}\right]\right)$, the
support is unbounded. Similar statements are true for
non-monotonically tailed PDFs with replacing $p_N(\cdot)$ by well
chosen lower and upper envelopes whose tails are
monotone. Furthermore, whenever some analyticity properties are
satisfied by the triplet, the discrete nature of the optimal
distribution is guaranteed. Discreteness holds also if additional
input compact support constraints are imposed.

These results are very broad; They are consistent with a multitude of
previously know capacity results, and provide solutions for a
multitude of new channel models. The generalization is one to many:
generic input-output functions, generic cost functions and generic
noise PDFs which include a large number of well-known noise models
such as the Gaussian, generalized-Gaussian, alpha-stable and their
mixtures. Interestingly, the results hold for all cost functions that
are $\omega\left(\ln |x|\right)$ where it is guaranteed that the
channel capacity exists and is finite.

The main idea behind the proofs of the theorems is the
characterization of the behavioral pattern of the KKT equation at
infinity after providing lower and upper bounds on some quantities of
interest. A key property is the subadditivity of the logarithm of the
inverse of a lower bound on the noise PDF at large values. This was
referred to as property C7 and is satisfied by all noise PDFs whose
tail has a dominant polynomial or exponential component.

A direct implication of the results concerns the question on what are
suitable power measures of the input signals of a communication
channel. Though the question seems to be absurd when dealing with
noise models with finite second moment where the natural power measure
would be the standard average power --which corresponds the cost
function $\mathcal{C}(|x|) = x^2$, defining a power measure when the
noise second moment is infinite is deemed crucial. This is due to the
fact that the natural signal-to-noise ratio would be equal to
zero. Based on our results, suitable average power measures should
correspond to cost functions which are ``at most''
$\Theta\left(\ln\left[\frac{1}{p_N\left(f(x)\right)}\right]\right)$
since otherwise capacity will be achieved by an input having a finite
power while the input space contains distributions having an infinite
one. Hence, each channel has its suitable average power measure
resulting from a suitable cost function. For example, a suitable cost
behaves ``at most'' like $\Theta(x^2)$ for the linear Gaussian channel
and has a logarithmic growth for channels with polynomially tailed
additive noise.

\appendices
\label{appen}

%

\section{Weak Differentiability of $I(\cdot)$ at $F^*$}
\label{ap:Prop_I}

\begin{theorem}
  Let $F^*$ be an optimal input distribution. Under a cost constraint
  $\int \cost{X}\,dF(x) \leq A$, $A > 0$, the mutual information
  $I(F)$ between the input and the output of
  channel~(\ref{eq:generic_ch1}) is weakly differentiable at $F^*$.
  \label{th:mutinf2}
\end{theorem}
Before proceeding to the proof, we note that the existence of an
optimal $F^*$ and the finiteness of the solution are insured as per
the discussion in Section~\ref{KKT}.

\begin{proof}
  Let $\theta$ be a number in $[0,1]$, $(F^*,F) \in \mathcal{P}_{A}
  \times \mathcal{P}_{A}$ and define $F_{\theta} = (1-\theta)F^* +
  \theta F$. The weak derivative of $I(.)$ at $F^*$ in the direction
  of $F$ is defined as,
  \begin{equation*}
    I'(F^*,F) \triangleq \lim_{\theta\to 0^+}\frac{I(F_{\theta})-I(F^*)}{\theta},
  \end{equation*}
  whenever the limit exists.
  For simplicity, we denote by
  \begin{equation*}
    t(x) = i(x;F^*),
  \end{equation*}
  where $i(x;F)$ is given by equation~(\ref{eq:i}),
  and we prove
  \begin{eqnarray*}
    I'(F^*,F) &=& -\int p(y;F)\ln p(y;F^*) \, dy \,\, - \,\, h_Y(F^*)\\
    &=& \int t(f(x)) \, dF(x) \,\,-\,\, h_Y(F^*),
  \end{eqnarray*}
  where by Tonelli, the interchange is valid as long as the integral
  term is finite which we prove next. 
  Using L'H\^{o}pital's rule,
  \begin{align}
    I'(F^*,F) = & \lim_{\theta\to 0^+}\frac{I(F_\theta)-I(F^*)}{\theta} = \lim_{\theta\to 0^+}\frac{h_Y(F_\theta)-h_Y(F^*)}{\theta} \nonumber\\
    = & \lim_{\theta \rightarrow{0^{+}}} -\left[\int p(y;F_\theta)
      \ln{p(y;F_\theta)}\,dy\right]^{'}, \label{eq:weakderiv}
  \end{align}
  where the derivative is with respect to $\theta$. In order to evaluate $\left[\int
    p(y;F_\theta)\ln{p(y;F_\theta)}\,dy\right]^{'}$ we use the
  definition of the derivative
  \begin{align}
    &\left[ \int p(y;F_\theta)\ln{p(y;F_\theta)}\,dy \right]^{'} \nonumber\\
    &= \lim_{h \rightarrow{0}} \left[ \frac{\int p(y;F_{\theta + h})\ln{p(y;F_{\theta + h})}\,dy}{h} - \frac{\int p(y;F_\theta)
        \ln{p(y;F_\theta)}\,dy}{h} \right], \nonumber
  \end{align}
  where by the limit we mean that both, the limit as $h$ goes to
  $0^{+}$ and the limit as $h$ goes to $0^{-}$ exist and are equal. In
  what follows, we only provide detailed evaluations as $h$ goes to
  $0^{+}$ since those when $h$ goes to $0^{-}$ are similar. Using the
  mean value theorem, for some $0 \leq{c(h)} \leq{h}$,
  \begin{align*}
    & \lim_{h \rightarrow{0+}} \left[ \frac{\int p(y;F_{\theta + h})\ln{p(y;F_{\theta + h})}\,dy}{h} - \frac{\int p(y;F_\theta)\ln{p(y;F_\theta)}\,dy}{h} \right] \\
    = & \lim_{h \rightarrow{0^{+}}} \int \left[ p(y;F_\theta)\ln{p(y;F_\theta)}\right]^{'} _{|_{\theta +c(h)}}\,dy.
  \end{align*}

  Now, since $p(y;F_\theta) = p(y;F^*) + \theta \left[p(y;F) -
    p(y;F^*)\right]$,
  \begin{align}
    &\lim_{h \rightarrow{0^{+}}} \int \left[ p(y;F_\theta)\ln{p(y;F_\theta)}
    \right]^{'}_{|_{\theta +c(h)}}\,dy \nonumber \\
    =& \lim_{h \rightarrow{0^{+}}} \int \left[p(y;F) - p(y;F^*)\right]
    \ln{p(y;F_{\theta + c(h)})} \,dy + \int\left[p(y;F) - p(y;F^*)\right]\,dy \nonumber \\
    =& \hspace{- 0.01cm} \int \lim_{h \rightarrow{0^{+}}} \left[p(y;F) - p(y;F^*)\right]
    \ln{p(y;F_{\theta + c(h)})} \, dy\label{eq:jusint}\\
    =& \int \left[p(y;F) - p(y;F^*)\right]\ln{p(y;F_\theta)} \,dy, \label{eq:contpc1} 
  \end{align}
  where~(\ref{eq:contpc1}) is due to the fact that $c(h)
  \rightarrow{0}$ as $h \rightarrow{0}$ and that $p(y;F_\theta)$ is
  continuous in $\theta$ by virtue of its linearity,
  and~(\ref{eq:jusint}) is due to DCT. Indeed,
  \begin{equation*}
    \left| \left[ p(y;F) - p(y;F^*) \right] \ln{p(y;F_{\theta + c(h)})} \right| \leq \bigl( p(y;F) + p(y;F^*) \bigr) \left|
      \ln{p(y;F_{\theta + c(h)})} \right|,
  \end{equation*}
  and
  \begin{align*}
    p(y;F_{\theta + c(h)}) & = \left[1 - \theta - c(h)\right]p(y;F^*) + \left[\theta + c(h)\right]p(y;F) \\
    & \geq{ \left[1 - \theta - c(h)\right]p(y;F^*) } \geq{ \frac{1}{2} \, p(y;F^*)},
  \end{align*}
  whenever $\theta + c(h) \leq{\frac{1}{2}}$, which is true since both
  $\theta$ and $c(h)$ are arbitrarily small.  Therefore, since $0 <
  p(y;F)< 1$ for all $F$
  \begin{equation*}
    \left| \left[ p(y;F) - p(y;F^*) \right] \ln{p(y;F_{\theta + c(h)})} \right| \leq - \bigl( p(y;F) + p(y;F^*) \bigr) 
    \ln{\left[\frac{1}{2} \, p(y;F^*)\right]}.
  \end{equation*}

  Since $h_{Y}(F) = - \int \,p(y;F)\ln{p(y;F)} \, dy$ is finite for
  all $F$ in $\mathcal{P}_{A}$~\cite[Theorem 2]{FAF15}, $-
  p(y;F^*)\ln{p(y;F^*)}$ is integrable. It remains to prove that $-
  p(y;F)\ln{p(y;F^*)}$ is integrable to justify~(\ref{eq:jusint}) and
  hence~(\ref{eq:contpc1}).  To this end, we will proceed by choosing
  first a specific $F(\cdot)$, namely
  \begin{equation*} 
    F_{s}(x) = \left( 1-\frac{B_s}{\mathcal{C}(x_s)} \right) u(x)\footnote[5]{where $u(x)$ 
      denotes the Heaviside unit step function.} + \frac{B_s}{\mathcal{C}(x_s)} u(x-x_{s}),
  \end{equation*} 
  for some $x_{s} > 0$ such that $\mathcal{C}(x_s) > 0$ and where $(0
  <) \,B_s < \min\,\{A;\mathcal{C}\left(x_s\right)\}$. We note that
  $F_{s} \in \mathcal{P}_{A}$ since $\mathcal{C}(0) = 0$ and hence
  $\int \mathcal{C}(|x|) dF_s = B_s \leq A$. If $F_s$ were the input
  distribution, it would induce the following output
  \begin{equation}
    p(y;F_{s}) =  \left( 1-\frac{B_s}{\mathcal{C}(x_s)} \right) p_{N}(y) 
    +\frac{B_s}{\mathcal{C}(x_s)} p_{N}(y - f(x_{s})).
    \label{eq:py_Fs}
  \end{equation}
  Equation~(\ref{eq:py_Fs}) along with lemma~\ref{lemlow} and
  properties C7 and C8 show that $-p(y;F_{s})\ln{p(y;F^*)}$ is
  integrable and~(\ref{eq:contpc1}) is justified for $F \equiv
  F_s$. Hence,
  \begin{eqnarray*}
    I'(F^*,F_{s}) &=& \lim_{\theta \rightarrow{0^{+}}} -\left[\int p(y;F^*_{\theta})
      \ln{p(y;F^*_{\theta})}\,dy\right]^{'}\\
    &=& \lim_{\theta \rightarrow{0^{+}}} \int \left[p(y;F_s) - p(y;F^*)\right]\ln{p(y;F^*_{\theta})} \,dy =  \int t(f(x))dF_{s}(x)\,\,-\,\,h_Y(F^*).
  \end{eqnarray*}
  where the interchange between the limit and integral sign is
  justified in an identical fashion as done to
  validate~(\ref{eq:contpc1}).

  Now, since $F^{*}$ is optimal, necessarily $I'(F^*,F_{s}) \leq{0}$
  (see Appendix C in~\cite{fahsj}), which implies that
  \begin{equation*}
    \int t(f(x)) \, dF_{s}(x) \leq h_Y(F^*).
  \end{equation*}
  
  Plugging in the expression of $F_{s}(x)$ yields,
  \begin{equation}
    \left(1-\frac{B_s}{\mathcal{C}(x_s)}\right)t(f(0)) + \frac{B_s}{\mathcal{C}(x_s)}t(f(x_{s})) \leq h_Y(F^*) \Leftrightarrow \quad  t(f(x_{s})) \leq \frac{h_Y(F^*) - t(f(0))}{B_s} \, \mathcal{C}(x_s) + t(f(0)). \label{uppco}
  \end{equation}
  The above equation is valid for any $x_s > 0$ (such that
  $\mathcal{C}(x_s) > 0$) and therefore for all $|x| \geq x_s$ since
  $\mathcal{C}(|x|)$ is non-decreasing in $|x|$. we proceed by writing
  \begin{equation*}
    \int t(f(x)) \, dF =  \int_{|x| \leq x_s} t\left(f(x)\right) \, dF + \int_{|x| > x_s} t(f(x)) \, dF.
  \end{equation*} 
  As for the first integral term, we have:
  \begin{align}
    &\int_{|x| \leq x_s} t\left(f(x)\right) \, dF\nonumber\\
    &= -\int_{|x| \leq x_s} \int  p_N(y-f(x)) \,\ln\,p(y;F^*)\,dy\, dF\nonumber\\
    &= -\int_{|x| \leq x_s} \int_{|y| \geq y_0} p_N(y-f(x)) \,\ln\,p(y;F^*)\,dy\, dF\nonumber\\
    &\quad -\int_{|x| \leq x_s} \int_{|y| \leq y_0} p_N(y-f(x)) \,\ln\,p(y;F^*)\,dy\, dF\label{lasteq0}
  \end{align}
  Using lemma~\ref{lemlow} and property C7, the first term of
  equation~(\ref{lasteq0}) is finite. As for the second term, it is
  finite by the fact that $p_(y;F^*)$ is positive and continuous hence
  achieves a positive minimum on compact subsets of $\Reals$. When it
  comes to the range $|x| > x_s$, we use the upper bound
  in~(\ref{uppco}) which gives:
  \begin{eqnarray*}
    \int_{|x| > x_s} t(f(x)) \, dF &\leq&  \int_{|x| > x_s}\left(\frac{h_Y(F^*) - t(f(0))}{B_u} \, \mathcal{C}(|x|) + t(0)\right)\,dF\\
    &\leq& \frac{h_Y(F^*) - t(f(0))}{B_u} A + t(f(0)),
  \end{eqnarray*}
  which is finite. 

  In conclusion,
  \begin{equation*}
    -\int p(y;F) \ln{p(y;F^*)} \,dy = \int t(f(x)) \, dF < \infty,
  \end{equation*}
  and $ I'(F^*,F) = \int t(f(x)) \,dF - h_Y(F^*)$, $\forall \, F \in
  \mathcal{P}_{A}$.
\end{proof}

\subsection*{Cost}
  The mapping from $\mathcal{F}$ to $\Reals$:
  \begin{equation*}
    \mathcal{T}(F) = \int\,\cost{x}\,dF - A
  \end{equation*}  
 is weakly differentiable on $\mathcal{P}_{A}$ as well. In fact,
  \begin{equation*}
    \mathcal{T}'(F^*,F) = \mathcal{T}(F) - \mathcal{T}(F^*),
  \end{equation*}          
  which is finite, since $-A < \mathcal{T}(F) \leq 0$ for all $F \in
  \mathcal{P}_{A}$.


\section{Rate of Decay of $S(\alpha,\beta,\gamma,\delta)$ on the
  Horizontal Strip}
\label{ratedec}

We study in this appendix the rate of decay of alpha-stable
distributions $S(\alpha,\beta,\gamma,\delta)$ on the horizontal strip
$\mathcal{S}_{\eta} = \{z \in \Complex : |\Im(z)| < \eta\}$ where
$\eta$ is a small-enough positive number.
 
We prove in this appendix that $|p_N(z)| = O
\left(\frac{1}{|\Re(z)+\delta|^{\alpha + 1}} \right)$ as $|\Re(z)|
\rightarrow \infty$, whenever $N \sim S(\alpha,\beta,\gamma,\delta)$
and $z \in \mathcal{S}_{\eta}$. The study is limited to the case:
$\alpha \in [1,2)$, $\beta \in ]-1,1[$, $\gamma \in \Reals^{+*}$ and
$\delta \in \Reals$.

Before we proceed, we first prove the following Lemma:
\begin{lemma}
  Whenever $N \sim S(\alpha, \beta, \gamma, \delta)$, where $\alpha
  \in [1,2)$, $\beta \in ]-1,1[$, $\gamma \in \Reals^{+*}$ and $\delta
  \in \Reals$, $p_N( \cdot )$ can be formally extended on
  $\mathcal{S}_{\eta} = \{z \in \Complex : |\Im(z)| < \eta\}$ as
  \begin{equation}
    p_N(z)  = \frac{1}{2 \pi}\int_{\Reals} e^{-izt} \phi(t) dt.
    \label{compext}
  \end{equation}
  \label{le:AS_Ext}
\end{lemma}
  
\begin{proof}
  By definition, 
  \begin{equation*}
    p_N(x)  = \frac{1}{2 \pi} \int_{-\infty}^{\infty} e^{-ixt} \phi(t) dt,
  \end{equation*}
  where 
  \begin{eqnarray*}
    \phi(t) & = & \exp \left[i \delta t - \asexp{t} \right] \\
    \Phi(t) & = & \left\{ \begin{array}{ll} 
        \displaystyle \tan \left(\frac{\pi \alpha}{2} \right) \quad & \alpha \neq 1\\
        \displaystyle  -\frac{2}{\pi} \ln|t| \quad & \alpha = 1.
      \end{array} \right.
  \end{eqnarray*}
  
  Let $p_N(z)$ be the extension of $p_N(x)$ on $\Complex$. It is known
  that $p_N(z)$ is analytic on $\mathcal{S}_{\eta}$ (see~\cite{zolo}
  for example) . Now, define
  \begin{equation*}
    q(z)  = \frac{1}{2 \pi}\int_{-\infty}^{\infty} e^{-izt} \phi(t) dt,
  \end{equation*}
  for all $z = (x + i y) \in \Complex$. If we establish that $q(z)$ is
  analytic on $\mathcal{S}_{\eta}$ then by the identity theorem,
  $p_N(z) = q(z)$, for all $z \in \mathcal{S}_{\eta}$. We start by
  proving the continuity of $q(z)$:
  \begin{align}
    \lim_{z \rightarrow z_0} q(z) = & \lim_{z \rightarrow z_0} \frac{1}{2 \pi} \int_{-\infty}^{\infty} e^{-izt} \phi(t) dt \nonumber\\
    = & \frac{1}{2 \pi} \int \lim_{z \rightarrow z_0}  e^{-izt} \phi(t) dt  \label{intermilan}\\
    = &  \frac{1}{2 \pi} \int e^{-iz_0t} \phi(t) dt  = q(z_0).\nonumber
  \end{align}
  where the interchange in~(\ref{intermilan}) is justified by DCT
  since:
  \begin{equation*}
    \left|e^{-izt} \phi(t)\right| \leq e^{yt-|\gamma t|^{\alpha}},
  \end{equation*}
  which is integrable on $\mathcal{S}_{\eta}$ since $\eta$ is
  small-enough and chosen so that $|y| < \eta \leq
  \gamma^{\alpha}$. Now, let $\Delta \subset \mathcal{S}_{\eta}$ be a
  compact triangle and denote by $\partial \Delta$ its boundary and
  $|\Delta|$ its perimeter. We obtain
  \begin{align}
    \int_{\partial \Delta} \hspace{-0.2cm}q(z)dz & = \frac{1}{2 \pi} \int_{\partial \Delta} \int_{\Reals}  
    e^{-izt} \phi(t) \,dt\,dz \nonumber\\
    &=  \frac{1}{2 \pi} \int_{\Reals}\int_{\partial \Delta} \hspace{-0.2cm}e^{-izt} \phi(t) \,dz\,dt \label{fub0} \\
    &=  \int_{\Reals} \phi(t)\,\int_{\partial \Delta} \hspace{-0.2cm}e^{-izt}dz = 0, \nonumber
  \end{align}
  where the last equation is due to the fact that $e^{-izt}$ is
  entire. The interchange in~(\ref{fub0}) is valid by Fubini since
  \begin{equation*}
    \frac{1}{2 \pi} \int_{\partial \Delta} \int_{\Reals} \left|e^{-izt} \phi(t)\right| \,dt\,dz \, 
    \leq \, \frac{1}{2 \pi} \int_{\partial \Delta} \int_{\Reals}  e^{yt-|\gamma t|^{\alpha}} \,dt\,dz \,
    < \, \frac{|\Delta|}{2 \pi} \, \int_{\Reals}  e^{yt-|\gamma t|^{\alpha}}| \,dt \,
    < \, \infty.
  \end{equation*} 

  By applying Morera's Theorem~\cite[sec. 53]{church}, $q(z)$ is
  analytic on $\mathcal{S}_{\eta}$ and the result is established.
\end{proof}

Note that equation~(\ref{compext}) shows that $p_N (z) = p_{N'}
(z-\delta)$ where $N' \sim S(\alpha, \beta, \gamma, 0)$. Therefore, and
without loss of generality, we restrict our analysis in the remainder
of this section to $p_N(z)$, for $N \sim S(\alpha, \beta, \gamma, 0)$.

For $z = (x + iy)$,
\begin{align}
  p_N(z) & = \frac{1}{2 \pi}\int_{-\infty}^{\infty} e^{-izt - \asexp{t}} dt 
  = \frac{1}{2 \pi}\int_{-\infty}^{\infty} e^{-ixt + yt - \asexp{t}} dt \nonumber\\
  &= \frac{1}{2 \pi}\int_{-\infty}^{\infty} e^{-ixt - \asexp{t}} \sum_{n=0}^{\infty}\frac{y^n}{n!} \, t^n \,dt \nonumber\\
  &= \frac{1}{2 \pi} \sum_{n=0}^{\infty}\frac{y^n}{n!} \int_{-\infty}^{\infty} t^n e^{-ixt - \asexp{t}} \,dt. \label{ext}
  \end{align}
The interchange in~(\ref{ext}) is justified by DCT. Indeed,
\begin{equation*}
  \left| \sum_{n=0}^N \frac{y^n}{n!}\, t^n e^{-ixt - \asexp{t}} \right| \leq \sum_{n=0}^{\infty}
  \frac{|y|^n}{n!} \, |t|^n e^{-|\gamma t|^{\alpha}} = e^{|y||t|-|\gamma t|^{\alpha}},
\end{equation*}
which is integrable for $|y| < \eta \, \, (\leq \gamma^{\alpha})$ and
$\alpha \geq 1$.
Now we proceed to studying the rate of decay in two separate cases.

\subsection{Rate of Decay for $1 <\alpha <2$:}

In this case $\Phi(t)$ is a constant and it is equal to $\Phi(t) =
\Phi = \tan \left(\frac{\pi \alpha}{2} \right)$. Then, using equation~(\ref{ext}), we obtain by the change of
variable $u = \gamma t$
\begin{align}
 p_N(z) &= \frac{1}{2 \pi \gamma} \sum_{n=0}^{\infty}\frac{1}{n!} \left(\frac{y}{\gamma}\right)^n \int_{-\infty}^{\infty} t^n
  e^{-i\frac{x}{\gamma}t - \left[1-i\beta\sgn(t)\Phi\right]|t|^{\alpha}} \,dt \nonumber\\
  &= \frac{1}{2 \pi \gamma} \sum_{n=0}^{\infty}\frac{1}{n!} \left(\frac{y}{\gamma}\right)^n T_n \left( -\frac{x}{\gamma}; \beta \right), \label{ext1}
\end{align}
where $T_n(x;\beta)$ is a function defined as $\displaystyle
T_n(x;\beta) \eqdef \int_{-\infty}^{\infty} t^n e^{ixt -
  \left[1-i\beta\sgn(t)\Phi\right]|t|^{\alpha}} dt$.\footnote{Note
  that $T_n(-x;\beta) = (-1)^n T_n(x;-\beta)$ and that $p^{(n)}_N(x) =
  \frac{1}{2 \pi} \frac{(-i)^{n}}{
    \gamma^{n+1}}\,T_n(-\frac{x}{\gamma};\beta) = \frac{1}{2
    \pi}\frac{i^n}{\gamma^{n+1}}\,T_n(\frac{x}{\gamma};-\beta)$, $n
  \in \Naturals^*$.}
Define $k_1 = (1 - i\beta \Phi)$ and denote by $\overline{k}_1 = (1 + i\beta \Phi)$ its
conjugate. In what follows, we study the behavior of the function
$T_n(x;\beta)$.

For $n \geq 1$ and $x >0$, we have
\begin{align}
  x^{n+\alpha+1} T_n(x;\beta) = & \, x^{n+\alpha} \left[ \int_{0}^{\infty}  x t^n e^{ixt - k_1t^{\alpha}} dt 
    + (-1)^n  \int_{0}^{\infty}  x t^n e^{-ixt - \overline{k}_1t^{\alpha}} dt \right] \nonumber\\
  = & -i \, x^{n+\alpha} \left[ \int_{0}^{\infty} t^n\left(e^{ixt -k_1t^{\alpha}}\right)^{'}\,dt + k_1 \, \alpha \int_{0}^{\infty} t^{n + \alpha -1} 
    e^{ixt-k_1t^{\alpha}} dt \right.  \nonumber\\
  & \qquad  + (-1)^{n-1} \left. \int_{0}^{\infty} t^n\left(e^{-ixt -\overline{k}_1t^{\alpha}}\right)^{'}\,dt + (-1)^{n-1}\overline{k}_1 \, 
    \alpha \int_{0}^{\infty} t^{n + \alpha -1} e^{-ixt-\overline{k}_1t^{\alpha}} dt \right] \nonumber\\
  = &\, i n x^{n+\alpha} \left[ \int_{0}^{\infty} t^{n-1} e^{ixt-k_1t^{\alpha}} dt + (-1)^{n-1}\int_{0}^{\infty} t^{n-1} e^{-ixt-\overline{k}_1t^{\alpha}} dt\right] \nonumber\\
  & \quad - i \alpha x^{n+\alpha} \left[ k_1 \int_{0}^{\infty}  t^{n + \alpha -1} e^{ixt-k_1t^{\alpha}} dt + (-1)^{n-1}\overline{k}_1\int_{0}^{\infty} t^{n + \alpha -1} 
    e^{-ixt-\overline{k}_1t^{\alpha}} dt \right] \label{byparts} \\
  = & \, i \, n \, x^{n+\alpha} T_{n-1}(x;\beta) - i \alpha \biggl[ k_1S_{n}(x;k_1) 
  + (-1)^{n-1}\overline{k}_1\overline{S_{n}}(x;\overline{k}_1) \biggr], \label{init}
\end{align} 
where equation~(\ref{byparts}) is obtained by integration by parts and
regrouping, and where $\overline{S_{n}}(\cdot;\cdot)$ is the complex
conjugate of $S_{n}(\cdot;\cdot)$ defined as,
\begin{equation*}
  S_{n}(x; k_1) = x^{n+\alpha} \int_{0}^{\infty} t^{n + \alpha -1} e^{ixt-k_1t^{\alpha}} dt
  = c \int_{0}^{\infty} e^{iv^{c}-k_1\zeta v^{\alpha c}} dv,
\end{equation*}
where $c = \frac{1}{n+\alpha} \, (> 0)$, $\zeta = x^{-\alpha} \, (>
0)$ and the change of variable is $v=(xt)^{n+\alpha}$. As $x \to
\infty$, $\zeta \to 0^{+}$ and hence
\begin{align}
  \lim_{x \to +\infty} S_{n}(x; k_1) & = c \lim_{\zeta \rightarrow 0^+} \int_{0}^{\infty} e^{iv^{c}-k_1\zeta v^{\alpha c}} \,dv 
  = c \lim_{\zeta \rightarrow 0^+} \int_{0}^{\infty} \lim_{\theta \rightarrow 0} e^{iv^{c}e^{ic \theta}-k_1\zeta 
    v^{\alpha c}e^{i \alpha c \theta} + i \theta}\,dv \nonumber\\
  &= c \lim_{\zeta \rightarrow 0^+} \lim_{\theta \rightarrow 0} \int_{0}^{\infty}  e^{iv^{c}e^{ic \theta}-k_1\zeta 
    v^{\alpha c}e^{i \alpha c \theta} +i\theta}\,dv \label{interchange1}\\
  &= c  \lim_{\theta \rightarrow 0} \lim_{\zeta \rightarrow 0^+} \int_{0}^{\infty}  e^{iv^{c}e^{ic \theta}-k_1\zeta 
    v^{\alpha c}e^{i \alpha c \theta} +i\theta}\,dv \label{interchange2}\\
  &= c \lim_{\theta \rightarrow 0} \int_{0}^{\infty}  e^{iv^{c}e^{ic \theta} + i\theta} \,dv \label{interchange3}\\
  &= c \lim_{\theta \rightarrow 0} \, \, \lim_{R \rightarrow \infty, \rho \rightarrow 0} \int_{\text{L}_1} e^{iz^{c}} \,dz, \nonumber
\end{align}
where $z=ve^{i \theta}$ and $\text{L}_1 = \{z \in \Complex : z = ve^{i
  \theta}, \, 0 < \rho \leq v \leq R
\}$. Equation~(\ref{interchange1}) is justified by DCT since:
\begin{equation*}
  \left| e^{iv^{c}e^{ic \theta}-k_1\zeta v^{\alpha c}e^{i \alpha c \theta} +i\theta} \right| \leq e^{-v^{c}\sin(c \theta)-\zeta 
    v^{\alpha c}\left[\cos(\alpha c \theta) + \beta \Phi \sin (\alpha c \theta)\right]} \leq e^{-\frac{\zeta}{2} v^{\alpha c}},
\end{equation*}
for small-enough $\theta$, and the upper-bound is integrable since $c$
and $\zeta$ are positive. The last inequality is justified by virtue
that $\sin(c \theta) > 0$ and $\left[\cos(\alpha c \theta) + \beta
  \tan \frac{\alpha \pi}{2} \sin (\alpha c \theta)\right] >
\frac{1}{2}$ for small positive $\theta$.
Similarly,~(\ref{interchange3}) is justified because the integrand
in~(\ref{interchange2}) is $O(e^{-v^{c} \sin c \theta})$ as $\zeta
\rightarrow 0^+$ which is also integrable. The interchange between the
two limits in~(\ref{interchange2}) is valid by the preceding argument
as long as the result in~(\ref{interchange3}) is finite. To evaluate
the limit of $\int_{\text{L}_1} e^{iz^{c}} dz$ as $R \rightarrow
\infty$, $\rho \rightarrow 0$, we use contour integration over
$\mathcal{C}$ shown in Figure~\ref{fig:curve}.

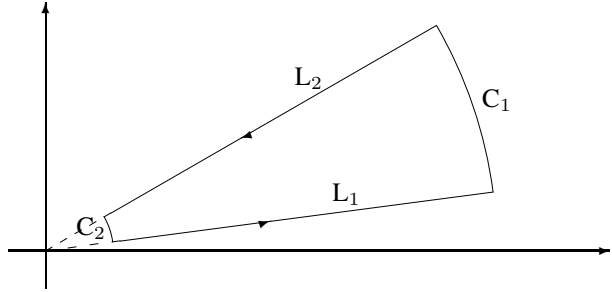
\begin{figure}[!htb]
  \begin{center}
    \setlength{\unitlength}{1cm}
    
    \begin{picture}(8,3.8)(-0.5,-0.5)
      \put(-0.5,0){\vector(1,0){8}}
      \put(0,-0.5){\vector(0,1){3.8}}
      
      \drawline(0.99144,0.13053)(5.9487,0.78316)
      \put(2.9743,0.39158){\vector(3,1){0}}
      \drawline(0.86603,0.50000)(5.1962,3.0000)
      \put(2.5981,1.5000){\vector(-3,-1){0}}
      \put(4,0.75){\makebox(0,0){$\text{L}_1$}}
      \put(3.5,2.3){\makebox(0,0){$\text{L}_2$}}
      
      \dashline{0.125}(0,0)(0.99144,0.13053)
      \dashline{0.125}(0,0)(0.86603,0.50000)
      
      \arc{12}{5.7596}{6.1523}
      \put(6,2){\makebox(0,0){$\text{C}_1$}}
      \arc{1.8}{5.7496}{6.1523}
      \put(0.8,0.3){\makebox(0,0)[r]{$\text{C}_2$}}
    \end{picture}
    \caption{The contour $\mathcal{C}$. \label{fig:curve}}
  \end{center}
\end{figure}

The arcs $\text{C}_1$ and $\text{C}_2$ are of radius $R$, and $\rho$
respectively and are between angles $\theta$ and $\varphi
\buildrel\triangle\over = \frac{\pi}{2 c} \bmod 2\pi$. Note that
since we are interested in the limit as $\theta$ goes to zero, we can
always choose it small enough in order to have the contour
counter-clockwise. Finally, $\text{L}_2$ is a line connecting the
extremities of the arcs.

Now since $f(z) = e^{iz^{c}}$ is analytic on
and inside $\mathcal{C}$ (by choosing an appropriate branch cut in the
plane), by Cauchy's Theorem~\cite[p.111 Sec.2.2]{bc},
\begin{equation*}
  0 = \oint_{\mathcal{C}} f(z) \,dz = \int_{\text{L}_1} \hspace{-3pt} f(z) + \int_{\text{C}_1} \hspace{-3pt} f(z)
  + \int_{\text{L}_2} \hspace{-3pt} f(z) + \int_{\text{C}_2} \hspace{-3pt} f(z).
\end{equation*}

On ${\text{C}_1}$, we have:
\begin{align*}
  &\lim_{R \rightarrow \infty} \left| \int_{\text{C}_1} f(z) dz \right|
  =\lim_{R \rightarrow \infty} \left| \int_{\theta}^{\varphi} i R e^{i \phi} e^{iR^{c} e^{i c \phi}} \, d\phi \right|
  \leq \lim_{R \rightarrow \infty}\int_{\theta}^{\varphi}  R e^{-R^{c} \sin(c \phi)} d\phi
  = \int_{\theta}^{\varphi}  \lim_{R \rightarrow \infty} R e^{-R^{c} \sin(c \phi)} d\phi = 0,
\end{align*}
where the interchange is valid because $R e^{-R^{c} \sin(c \phi)}$ is
decreasing as $0<c \theta \leq c \phi \leq \frac{\pi}{2}$. 
Similarly, on ${\text{C}_2}$,
\begin{align*}
  &\lim_{\rho \rightarrow 0} \left| \int_{\text{C}_2} f(z) dz \right|
  =\lim_{\rho \rightarrow 0} \left| \int_{\theta}^{\varphi} i \rho e^{i \phi} e^{i\rho^{c}e^{ic \phi}} \, d\phi \right|
  \leq \lim_{\rho \rightarrow 0}\int_{\theta}^{\varphi} \rho e^{-\rho^{c} \sin(c \phi)} d\phi
  = \int_{\theta}^{\varphi}  \lim_{\rho \rightarrow 0} \rho e^{-\rho^{c} \sin(c \phi)} d\phi = 0,
\end{align*}
where we justify the interchange by virtue of the fact that $\rho
e^{-\rho^{c} \sin(c \phi)}$ is bounded for small values of $\rho$. It
remains to evaluate the integral on ${\text{L}_2}$ where $z=te^{i
  \frac{\pi}{2 c}}$,
\begin{equation*}
  \lim_{R \rightarrow \infty, \rho \rightarrow 0}  \, \int_{\text{L}_2} f(z) dz  = -\int_{0}^{\infty} e^{i \frac{\pi}{2 c}} e^{it^{c}e^{i \frac{\pi}{2}}} \, dt
  = -e^{i \frac{\pi}{2 c}} \int_{0}^{\infty}  e^{-t^{c}} \, dt = -e^{i \frac{\pi}{2 c}} \,\frac{1}{c}  \, \Gamma \left(\frac{1}{c}\right).
\end{equation*}
In conclusion,
\begin{equation*}
  \lim_{R \rightarrow \infty, \rho \rightarrow 0} \int_{\text{L}_1} f(z) \,dz = e^{i \frac{\pi}{2 c}} \,\frac{1}{c}\, \Gamma \left(\frac{1}{c}\right),
\end{equation*}
which implies that
\begin{equation*}
  \lim_{x \to +\infty} S_{n}(x; k_1) =  e^{i \frac{\pi}{2}(n+\alpha)} \Gamma(n+\alpha),
\end{equation*}
and by~(\ref{init}), we can write for $n \geq 1$
\begin{multline*}
  \lim_{x \rightarrow +\infty} \left[ x^{n+\alpha +1} T_n(x;\beta) - i \,n x^{n+\alpha} T_{n-1}(x;\beta) \right] \\
  = W_n(\beta) \eqdef -i \alpha \Gamma(n+\alpha) \biggl[ k_1e^{i \frac{\pi}{2}(n+\alpha)}  + (-1)^{n-1} \overline{k}_1 
 e^{-i \frac{\pi}{2}(n+\alpha)} \biggr],
\end{multline*}
which implies that $\displaystyle U_n(\beta) \eqdef \lim_{x
  \rightarrow +\infty} x^{n+\alpha +1} T_n(x;\beta)$ is a well defined
quantity because
\begin{equation*}
  U_0(\beta) = \lim_{x \rightarrow +\infty} \left[x^{\alpha +1} T_0(x;\beta)\right] = 2 \pi \gamma \lim_{x \rightarrow +\infty} \left[x^{\alpha +1} p_N(-\gamma x) \right],
\end{equation*}
exists --and is non zero for $\beta \neq 1$ and $U_0(1) = 0$, and
\begin{equation*}
  U_n(\beta) = \, inU_{n-1}(\beta) + W_n(\beta) 
  = n! \left[ i^n \, U_0(\beta) + \sum_{k=0}^{n-1}\frac{i^k}{(n-k)!}W_{n-k}(\beta)\right].
\end{equation*}
Furthermore, for $n \geq 0$,
\begin{eqnarray}
  \left| U_n(\beta) \right| & \leq & n! \left[ \left| U_0(\beta) \right| + \sum_{k=0}^{n-1} \frac{ \left| W_{n-k}(\beta) \right| } {(n-k)!}\right]
  \, \leq \,  n! \left[ \left| U_0(\beta) \right| + 2\alpha|k_1| \sum_{k=0}^{n-1} \frac{\Gamma(n +\alpha -k)}{(n-k)!} \right] \nonumber\\
  & \leq & n! \left[ \left| U_0(\beta) \right| + 4|k_1| \sum_{k=0}^{n-1} \frac{\Gamma(n+2-k)}{(n-k)!} \right] \label{tttt}\\
  & = & n! \left[ \left| U_0(\beta) \right| + 4 |k_1|\sum_{k=0}^{n-1} (n+1-k) \right] = 2\,n! \, \left(|k_1|n^2 + 3|k_1|n 
    + \frac{\left| U_0(\beta) \right|}{2}   \right)\nonumber,
\end{eqnarray} 
where equation~(\ref{tttt}) is justified using the fact that
$0 < \alpha < 2$ and $\Gamma(\alpha + l)$ is increasing in $\alpha > 0$ for $l \in
\Naturals^*$.

Now using equation~(\ref{ext1}),
\begin{align}
  & \lim_{x \rightarrow \infty} x^{\alpha + 1}\, \left| p_N(z) \right| 
  = \frac{1}{2 \pi \gamma} \lim_{x \rightarrow \infty}x^{\alpha +1} \left|  \sum_{n=0}^{\infty}\frac{1}{n!} \left(\frac{y}{\gamma}\right)^n 
    T_n \left( -\frac{x}{\gamma}; \beta \right) \right| \nonumber\\
  &=\frac{1}{2 \pi \gamma} \left| \sum_{n=0}^{\infty} \frac{1}{n!}\left(\frac{y}{\gamma}\right)^n \lim_{x \rightarrow \infty} x^{\alpha +1}
    T_n\left(-\frac{x}{\gamma};\beta \right) \right| \label{interchange4}\\
  &\leq \frac{1}{2 \pi \gamma} \sum_{n=0}^{\infty}\frac{1}{n!}\left|\frac{y}{\gamma}\right|^n\lim_{x
    \rightarrow \infty}x^{\alpha +1}\left| T_n\left(\frac{x}{\gamma};-\beta\right) \right| \nonumber \\
  &\leq \frac{1}{2 \pi \gamma} \sum_{n=0}^{\infty}\frac{1}{n!}\left|\frac{y}{\gamma}\right|^n\lim_{x
    \rightarrow \infty}x^{n+\alpha +1}\left|T_n\left(\frac{x}{\gamma};-\beta\right)\right| 
  =\frac{1}{2 \pi \gamma} \sum_{n=0}^{\infty}\frac{1}{n!}\left|\frac{y}{\gamma}\right|^n \gamma^{n + \alpha +1} \left|U_n(-\beta)\right| \nonumber\\
  &\leq \frac{\gamma^{\alpha}}{ \pi } \sum_{n=0}^{\infty}|y|^n  \left(|k_1|n^2 + 3|k_1|n + \frac{\left|U_0(-\beta)\right|}{2} \right), \nonumber
\end{align} 
which is finite because $|y| < \eta$ which is small-enough (and
assumed to be less than one), and where we used the fact that $f(x) =
|x|$ is continuous. The interchange in~(\ref{interchange4}) is valid
because the end result is finite.

In conclusion, $\displaystyle \lim_{x \rightarrow _+\infty} x^{\alpha
  + 1}\,|p_N(z)| < \infty$ which concludes our proof. 


\subsection{Rate of Decay for $\alpha = 1$:}

In this case, $\Phi(t) = -\frac{2}{\pi} \log|t|$ is a function of
$t$. According to equation~(\ref{ext}) and for $z = x + iy$,
\begin{equation}    
  p_N(z) = \frac{1}{2 \pi} \sum_{n=0}^{\infty}\frac{y^n}{n!} \int_{-\infty}^{\infty} t^n e^{-ixt-\gamma \bigl[1 - i \beta \sgn(t) \Phi (t) \bigr] |t|} \,dt.
  \label{extalpha1}
\end{equation}

Once more, we study the behavior of the integral $I_n(x) =
\int_{-\infty}^{\infty} t^n e^{-ixt-\gamma \left(1 - i \beta \sgn(t)
    \Phi \right) |t|} \,dt$ as
$x \rightarrow \infty$ for $n \geq 0$. We note that $I_0(x)= 2 \pi
p_N(x;1,\beta,\gamma,0)$ which is $\Theta\left(\frac{1}{x^2}\right)$. 
For $n \geq 1$,
\begin{align}
  I_n(x) & = \int_{-\infty}^{+\infty} t^n e^{-ixt-\gamma \left(1 - i \beta \sgn(t) \Phi \right) |t|} \,dt \nonumber\\
  &= \int_{0}^{+\infty} t^n e^{-ixt- \gamma \left(1 + i \frac{2}{\pi} \beta \log(t) \right) t} \,dt + \int_{-\infty}^{0} t^n e^{-ixt- \gamma \left(1 - i \frac{2}{\pi} \beta
      \log(-t) \right) (-t)} \,dt \nonumber\\
  &= \int_{0}^{+\infty} e^{-ixt}\, t^n e^{- \gamma \left(1 + i \frac{2}{\pi} \beta \log(t) \right) t} \,dt + (-1)^n \int_{0}^{+\infty} e^{ixt}\,t^n e^{- \gamma 
    \left(1 - i \frac{2}{\pi} \beta  \log(t) \right) t} \,dt \nonumber\\
  &=\left[-\frac{1}{ix}e^{-ixt}\,t^n e^{- \gamma \left(1 + i \frac{2}{\pi} \beta \log(t) \right) t}\right]_{0}^{+\infty}  +(-1)^n \left[\frac{1}{ix}e^{ixt}\,
    t^n e^{- \gamma \left(1 - i \frac{2}{\pi} \beta \log(t) \right) t} \right]_{0}^{+\infty} \nonumber\\ 
  &\quad + \frac{1}{ix}\int_{0}^{+\infty}e^{-ixt}\, \bigl[ nt^{n-1}-\gamma t^n - i\frac{2}{\pi}\beta \gamma t^n - i\frac{2}{\pi}\beta \gamma t^n \log(t) \bigr] 
  \,e^{- \gamma \left(1 + i \frac{2}{\pi} \beta \log(t) \right) t}\,dt \nonumber\\
  &\quad + \frac{(-1)^{n+1}}{ix}\int_{0}^{+\infty}e^{ixt}\,\bigl[nt^{n-1}-\gamma t^n + i\frac{2}{\pi}\beta \gamma t^n + i\frac{2}{\pi}\beta \gamma 
  t^n \log(t) \bigr] \,e^{- \gamma \left(1 - i \frac{2}{\pi} \beta \log(t) \right) t}\,dt \label{bypartsE1}\\
  &= \frac{1}{ix}\int_{0}^{+\infty}e^{-ixt}\, \bigl[ nt^{n-1}-\gamma t^n - i\frac{2}{\pi}\beta \gamma t^n - i\frac{2}{\pi}\beta \gamma t^n \log(t) \bigr]\, 
  e^{- \gamma \left(1 + i \frac{2}{\pi} \beta \log(t) \right) t}\,dt \nonumber\\
  &\quad +\frac{(-1)^{n+1}}{ix}\int_{0}^{+\infty}e^{ixt}\,\bigl[nt^{n-1}-\gamma t^n + i\frac{2}{\pi}\beta \gamma t^n + i\frac{2}{\pi}\beta \gamma t^n 
  \log(t) \bigr] \,e^{- \gamma \left(1 - i \frac{2}{\pi} \beta \log(t) \right) t}\,dt \nonumber\\
  &=\left[\frac{1}{x^2}e^{-ixt}\,\left[ nt^{n-1}-\gamma t^n - i\frac{2}{\pi}\beta \gamma t^n - i\frac{2}{\pi}\beta \gamma t^n \log(t) \right] \, 
    e^{- \gamma \left(1 + i \frac{2}{\pi} \beta \log(t) \right) t}\right]_{0}^{+\infty}\nonumber\\
  &\qquad \qquad \qquad \qquad \qquad \qquad -\frac{1}{x^2}\int_{0}^{+\infty}e^{-ixt}g_n(t)\,dt\nonumber\\
  &\quad +(-1)^{n+1} \left[-\frac{1}{x^2}e^{ixt}\,\left[ nt^{n-1}-\gamma t^n + i\frac{2}{\pi}\beta \gamma t^n + i\frac{2}{\pi}\beta \gamma t^n \log(t) \right]
    \, e^{- \gamma \left(1 - i \frac{2}{\pi} \beta \log(t) \right) t}\right]_{0}^{+\infty}\nonumber\\ 
  &\qquad \qquad \qquad \qquad \qquad \qquad + \frac{(-1)^{n+1}}{x^2}\int_{0}^{+\infty}e^{ixt}h_n(t)\,dt \label{bypartsE2}\\
  &=\frac{1}{x^2}\left((-1)^{n+1}\int_{0}^{+\infty}e^{ixt}\,h_n(t)\,dt - \int_{0}^{+\infty}e^{-ixt}g_n(t)\,dt\right)\label{decEa1}
\end{align}
where equations~(\ref{bypartsE1}) and~(\ref{bypartsE2}) are due to
integration by parts. The functions $g_n(\cdot)$ and $h_n(\cdot)$, $n
\geq 1$ are defined on $\Reals^{+*}$ and are given by:

\begin{align}
  &g_n(t) = \left[n(n-1)t^{n-2} -2n \gamma t^{n-1} + (\gamma^2 - \frac{4}{\pi^2}\beta^2 \gamma^2 + i \frac{4}{\pi}\beta \gamma^2)t^n 
    + (-\frac{8}{\pi^2}\beta^2 \gamma^2 + i \frac{4}{\pi}\beta \gamma^2)t^n \log(t) \right. \nonumber\\
  &\quad \quad  \left. -i\frac{2}{\pi}(2n+1)\beta \gamma t^{n-1} - i\frac{4}{\pi}n\beta \gamma t^{n-1}\log(t) - \frac{4}{\pi^2}\beta^2 
    \gamma^2 t^n \log^2(t)\right] e^{- \gamma \left(1 + i \frac{2}{\pi} \beta \log(t) \right) t}.\label{gexp} 
\end{align}
The term $n(n-1)t^{n-2}$ is equal to zero when $n = 1$ and $h_n(t)$ is
deduced from $g_n(t)$ by replacing $\beta$ by $-\beta$. The functions
$g_n(t)$, $h_n(t)$ are $\mathbb{L}^{1}(\Reals^+)$ functions and hence
by Riemann-Lebesgue~\cite[p.3 sec.2 th.1]{bochner49} their
$\mathbb{L}^{1}(\Reals^+)$ Fourier transforms are $o(1)$. Therefore
equation~(\ref{decEa1}) is $o(\frac{1}{x^2})$. Equivalently, $I_n(x) =
o(\frac{1}{x^2})$ as $x \rightarrow \infty$ for all $n \geq 1$. Now
using equation~(\ref{extalpha1}) we obtain:
\begin{align}    
  &\lim_{x \rightarrow \infty} 2\pi x^2|p_N(z)|\nonumber\\
  &= \lim_{x \rightarrow \infty}  \left|\sum_{n=0}^{\infty}\frac{y^n}{n!}x^2I_n(x)\right| \nonumber\\
  &= \lim_{x \rightarrow \infty} \left|2 \pi x^2 p_N(x) + \sum_{n=1}^{\infty}\frac{y^n}{n!}\left((-1)^{n+1}\int_{0}^{+\infty}e^{ixt}h_n(t)\,dt 
      - \int_{0}^{+\infty}e^{-ixt}g_n(t)\,dt\right)\right| \nonumber\\
  &\leq \lim_{x \rightarrow \infty} 2 \pi x^2 p_N(x) + \lim_{x \rightarrow \infty}\sum_{n=1}^{\infty}\frac{|y|^n}{n!}\left(\int_{0}^{+\infty}|h_n(t)|\,dt 
    + \int_{0}^{+\infty}|g_n(t)|\,dt\right) \nonumber\\
  &=\lim_{x \rightarrow \infty} 2 \pi x^2 p_N(x) + \sum_{n=1}^{\infty}\frac{|y|^n}{n!} \int_{0}^{+\infty}(|h_n(t)|+|g_n(t)|) \, dt \nonumber\\  
  &= \lim_{x \rightarrow \infty} 2 \pi x^2 p_N(x) + \int_{0}^{+\infty} \sum_{n=1}^{\infty}\frac{|y|^n}{n!}(|h_n(t)|+|g_n(t)|) \, dt \label{oarg}
\end{align} 
The interchange in~(\ref{oarg}) is valid since:
\begin{align}
  & \sum_{n=1}^{\infty}\frac{|y|^n}{n!}(|h_n(t)|+|g_n(t)|)\nonumber\\
  &\leq \sum_{n=1}^{\infty}\frac{|y|^n}{n!}\left[A_1n(n-1)t^{n-2} + A_2t^n + A_3t^n |\log(t)| + A_4(2n+1)t^{n-1}\right.\nonumber\\
  &\qquad \qquad \qquad \qquad \qquad \qquad \left. + A_5 n t^{n-1}|\log(t)| + A_6 t^n \log^2(t)\right] e^{- \gamma t}\nonumber\\
  &\leq e^{- \gamma t} \sum_{n=1}^{\infty}\frac{|y|^n}{n!}\left[ A_1n(n-1)t^{n-2} + (A_2 + A_3 |\log(t)| + A_6\log^2(t))t^n \right.\nonumber\\
  &\qquad \qquad \qquad \qquad \qquad \qquad \quad \left. + n(3 A_4   + A_5  |\log(t)|)t^{n-1}\right] \label{simine}\\
  &\leq e^{- \gamma t}\left[A_1y^2e^{|y|t} + (A_2 + A_3|\log(t)|+A_6\log^2(t))(e^{|y|t}-1) +|y|(3A_4 + A_5)e^{|y|t} \right] \nonumber\\
  &\leq e^{- (\gamma - |y|) t}\left[A_2 + (3A_4 + A_5)|y| + A_1 y^2 + A_3|\log(t)|+A_6\log^2(t) \right] \nonumber
\end{align}
which is integrable on $[0,+\infty[$ since $|y| < \eta \, (<
\gamma)$. The $A_{i}$s, $1 \leq i \leq 6$ are positive constants
function of $\beta$, $\gamma$ and can be derived from the expression
of $g_n(t)$ (equation~\ref{gexp}) and from that of $h_n(t)$
accordingly after taking the norm of each term in those
expressions. To write equation~(\ref{simine}), we used the obvious
inequality $2n+1 \leq 3n$ whenever $n \geq 1$. Back to~(\ref{oarg}),
\begin{align}    
  \lim_{x \rightarrow \infty} 2\pi x^2|p_N(z)|
  &\leq \lim_{x \rightarrow \infty} 2 \pi x^2 p_N(x) + \int_{0}^{+\infty} \sum_{n=1}^{\infty}\frac{|y|^n}{n!}(|h_n(t)|+|g_n(t)|) \, dt \nonumber\\
  &\leq \lim_{x \rightarrow \infty} 2 \pi x^2 p_N(x) + \int_{0}^{+\infty}l(t)\,dt \nonumber
\end{align}
where $l(t) = e^{- (\gamma - |y|) t}\left[A_2 + (3A_4 + A_5)|y| + A_1
  y^2 + A_3|\log(t)|+A_6\log^2(t) \right]$. Since $\lim_{x \rightarrow
  \infty} 2 \pi x^2 p_N(x)$ and $\int_{0}^{+\infty}l(t)\,dt$ are both
finite and non zero when $|y| < \eta \, (< \gamma)$, then
$\displaystyle 0 \leq \lim_{x \rightarrow \infty} 2\pi x^2|p_N(z)| <
\infty$ and $|p_N(z)| = O \left( \frac{1}{|\Re(z)|^{2}} \right)$ as
$\Re(z) \rightarrow \infty$ whenever $z \in \mathcal{S}_{\eta}$.

\section*{Acknowledgments}
The authors would like to thank Professor Aslan Tchamkerten for suggesting the idea of the converse.


\bibliographystyle{IEEEtran}
\bibliography{paper-1}

\begin{thebibliography}{10}
\providecommand{\url}[1]{#1}
\csname url@rmstyle\endcsname
\providecommand{\newblock}{\relax}
\providecommand{\bibinfo}[2]{#2}
\providecommand\BIBentrySTDinterwordspacing{\spaceskip=0pt\relax}
\providecommand\BIBentryALTinterwordstretchfactor{4}
\providecommand\BIBentryALTinterwordspacing{\spaceskip=\fontdimen2\font plus
\BIBentryALTinterwordstretchfactor\fontdimen3\font minus
  \fontdimen4\font\relax}
\providecommand\BIBforeignlanguage[2]{{%
\expandafter\ifx\csname l@#1\endcsname\relax
\typeout{** WARNING: IEEEtran.bst: No hyphenation pattern has been}%
\typeout{** loaded for the language `#1'. Using the pattern for}%
\typeout{** the default language instead.}%
\else
\language=\csname l@#1\endcsname
\fi
#2}}

\bibitem{Sha48_1}
C.~E. Shannon, ``A mathematical theory of communication, part i,'' \emph{Bell
  Syst. Tech. J.}, vol.~27, pp. 379--423, 1948.

\bibitem{Sha48_2}
------, ``A mathematical theory of communication, part ii,'' \emph{Bell Syst.
  Tech. J.}, vol.~27, pp. 623--656, 1948.

\bibitem{IA01}
I.~Abou-Faycal, M.~D. Trott, and S.~Shamai, ``{The capacity of discrete-time
  memoryless Rayleigh-fading channels},'' \emph{Information Theory, IEEE
  Transactions on}, vol.~47, no.~4, pp. 1290--1301, May 2001.

\bibitem{MK04}
M.~Katz and S.~Shamai, ``{On the capacity-achieving distribution of the
  discrete-time noncoherent and partially coherent AWGN channels},''
  \emph{Information Theory, IEEE Transactions on}, vol.~50, no.~10, pp.
  2257--2270, October 2004.

\bibitem{Nuriyev2005}
{R. Nuriyev and A. Anastasopoulos}, ``{Capacity and coding for the
  block-independent noncoherent AWGN channel},'' \emph{IEEE Transactions on
  Information Theory}, vol.~51, no.~3, pp. 866--883, March 2005.

\bibitem{Chung06}
C.~Luo, ``Communication for wideband fading channels: on theory and practice,''
  Ph.D. dissertation, Massachusetts Institute of Technology, February 2006.

\bibitem{fahsj}
{J. Fahs, I. Abou-Faycal}, ``{Using Hermite bases in studying
  capacity-achieving distributions over AWGN channels},'' \emph{Information
  Theory, IEEE Transactions on}, vol.~58, no.~8, August 2012.

\bibitem{SMITH71}
J.~G. Smith, ``{The information capacity of peak and average power constrained
  scalar Gaussian channels},'' \emph{Inform. Contr.}, vol.~18, pp. 203--219,
  1971.

\bibitem{Zhang2011}
L.~Zhang and D.~Guo, ``{Capacity of Gaussian Channels with Duty Cycle and Power
  Constraints},'' in \emph{IEEE International Symposium on Information Theory},
  Saint Petersburg, Russia, 2011, pp. 424--428.

\bibitem{Das}
A.~Das, ``{Capacity-achieving distributions for non-Gaussian additive noise
  channels},'' \emph{in Proc. IEEE International Symposium on Information
  Theory}, p. 432, June 2000, sorrento, Italy.

\bibitem{Fahs2}
{J. Fahs, N. Ajeeb, and I. Abou-Faycal}, ``{The capacity of average power
  constrained additive non-Gaussian noise channels},'' in \emph{IEEE
  International Conference on Telecommunications}.\hskip 1em plus 0.5em minus
  0.4em\relax Beirut, Lebanon, April 2012.

\bibitem{Aslan}
A.~Tchamkerten, ``{On the Discreteness of Capacity-Achieving Distributions},''
  \emph{IEEE Transactions on Information Theory}, vol.~50, no.~11, pp.
  2773--2778, November 2004.

\bibitem{IA10}
{I. Abou-Faycal, J. Fahs}, ``{On the capacity of some deterministic non-linear
  channels subject to additive white Gaussian noise},'' in \emph{{IEEE 17th
  International Conference on Telecommunications (ICT)}}, Doha, Qatar, April
  2010, pp. 63--70.

\bibitem{Fahs12}
J.~Fahs and I.~Abou-Faycal, ``{On the capacity of additive white alpha-stable
  noise channels},'' in \emph{IEEE International Symposium on Information
  Theory}, Cambridge, MA, USA, 2012, pp. 294--298.

\bibitem{Fahs14}
{J. Fahs, I. Abou-Faycal}, ``On the single-user capacity of some multiple
  access channels,'' in \emph{The Eleventh International Symposium on Wireless
  Communication Systems}, Barcelona, Spain, August, 26-29 2014.

\bibitem{Anan96}
V.~Anantharam and S.~Verdu, ``Bits through queues,'' \emph{IEEE Transactions on
  Information Theory}, vol.~42, no.~1, pp. 4--18, January 1996.

\bibitem{Fahs14-1}
{J. Fahs, I. Abou-Faycal}, ``{A Cauchy input achieves the capacity of a Cauchy
  channel under a logarithmic constraint},'' in \emph{IEEE International
  Symposium on Information Theory}, Honolulu, HI, USA, June 29 - July 4 2014.

\bibitem{blu}
{R. S. Blum, R. J. Kozick, and B. M. Sadler}, ``{An adaptive spatial diversity
  receiver for non-Gaussian interference and noise},'' \emph{Signal Processing,
  IEEE Transactions on}, vol.~47, no.~8, pp. 2100--2111, Aug. 1999.

\bibitem{amir}
{A. Nasri, and R. Schober}, ``{Performance of BICM-SC and BICM-OFDM systems
  with diversity reception in non-Gaussian noise and interference},''
  \emph{Communications, IEEE Transactions on}, vol.~57, no.~11, pp. 3316--3327,
  Nov. 2009.

\bibitem{fiorina2006}
{J. Fiorina}, ``{A simple IR-UWB receiver adapted to multi-user
  interferences},'' in \emph{IEEE Globecom}, San Francisco, CA, 27 November - 1
  December 2006, pp. 1--4.

\bibitem{BSF08}
N.~Beaulieu, H.~Shao, and J.~Fiorina, ``{P-order metric UWB receiver structures
  with superior performance},'' \emph{Communications, IEEE Transactions on},
  vol.~56, no.~10, pp. 1666--1676, October 2008.

\bibitem{stuck}
B.~W. Stuck and B.~Kleiner, ``{A statistical analysis of telephone noise},''
  \emph{Bell Syst. Tech. J.}, vol.~53, no.~7, pp. 1263--1320, 1974.

\bibitem{geo}
{P. G. Georgiou, P. Tsakalides, and C. Kyriakakis}, ``{Alpha-stable modeling of
  noise and robust time-delay estimation in the presence of impulsive noise},''
  \emph{Multimedia, IEEE Transactions on}, vol.~1, no.~3, pp. 291--301, 1999.

\bibitem{sousa92}
E.~S. Sousa, ``{Performance of a spread spectrum packet radio network link in a
  Poisson field of interferers},'' \emph{Information Theory, IEEE Transactions
  on}, vol.~38, no.~6, pp. 1743--1754, Nov. 1992.

\bibitem{jacek}
J.~Ilow and D.~Hatzinakos, ``{Analytic alpha-stable noise modeling in a Poisson
  field of interferers or scatterers},'' \emph{Signal Processing, IEEE
  Transactions on}, vol.~46, no.~6, pp. 1601--1611, Jun. 1998.

\bibitem{Win}
M.~Win, P.~Pinto, and L.~Shepp, ``A mathematical theory of network interference
  and its applications,'' \emph{Proceedings of the IEEE}, vol.~97, no.~2, pp.
  205 --230, February 2009.

\bibitem{Beaulieu}
N.~Beaulieu and D.~Young, ``Designing time-hopping ultrawide bandwidth
  receivers for multiuser interference environments,'' \emph{Proceedings of the
  IEEE}, vol.~97, no.~2, pp. 255 --284, February 2009.

\bibitem{marcel}
M.~Nassar, K.~Gulati, A.~Sujeeth, N.~Aghasadeghi, B.~Evans, and K.~Tinsley,
  ``Mitigating near-field interference in laptop embedded wireless
  transceivers,'' in \emph{IEEE International Conference on Acoustics, Speech
  and Signal Processing}, Las Vegas, NV, 30 March - 4 April 2008, pp. 1405
  --1408.

\bibitem{ElGhannoudi}
H.~El~Ghannudi, L.~Clavier, N.~Azzaoui, F.~Septier, and P.-a. Rolland,
  ``{Stable interference modeling and Cauchy receiver for an IR-UWB ad hoc
  network},'' \emph{Communications, IEEE Transactions on}, vol.~58, no.~6, pp.
  1748 --1757, Jun. 2010.

\bibitem{raj}
{A. Rajan and C. Tepedelenlio˘glu}, ``{Diversity combining over Rayleigh
  fading channels with symmetric alpha-stable noise},'' \emph{Communications,
  IEEE Transactions on}, vol.~9, no.~9, pp. 2968--2976, Sep. 2010.

\bibitem{Middle86}
{D. Middleton and A. D. Spaulding}, ``{A tutorial review of elements of weak
  signal detection in non-Gaussian EMI environments},'' {U.S. Dept. of
  Commerce},'' {NTIA Rep. 86-194}, 1986.

\bibitem{Kim98}
{Y. Kim and G. T. Zhou}, ``{The Middleton class B model and its mixture
  representation},'' Center for Signal and Image Processing, Georgia Institute
  of Technology, Atlanta, GA, Tech. Rep. CSIP TR-98-01, May 1998.

\bibitem{Shao}
M.~Shao and C.~Nikias, ``Signal processing with fractional lower order moments:
  stable processes and their applications,'' \emph{Proceedings of the IEEE},
  vol.~81, no.~7, pp. 986 --1010, Jul. 1993.

\bibitem{tsa}
{P. Tsakalides and C. L. Nikias}, ``{Maximumm likelihood localization of
  sources in noise modeled as a stable process},'' \emph{Signal Processing,
  IEEE Transactions on}, vol.~43, no.~11, pp. 2700--2713, Nov. 1995.

\bibitem{man}
{X. Ma and C. L. Nikias}, ``{Joint estimation of time delay and frequency delay
  in impulsive noise using fractional lower order statistics},'' \emph{Signal
  Processing, IEEE Transactions on}, vol.~44, no.~11, pp. 2669--2687, Nov.
  1996.

\bibitem{tsi}
{G. A. Tsihrintzis and C. L. Nikias}, ``{Performance of optimum and suboptimum
  receivers in the presence of impulsive noise modeled as an alpha-stable
  process},'' \emph{Communications, IEEE Transactions on}, vol.~43, no. 2/3/4,
  pp. 904--914, Feb./Mar./Apr. 1995.

\bibitem{gonza2}
{J. G. Gonzalez, J. L. Paredes, and G. R. Arce}, ``{Zero-order statistics: A
  mathematical framework for the processing and characterization of very
  impulsive signals},'' \emph{Signal Processing, IEEE Transactions on},
  vol.~54, no.~10, pp. 3839--3851, Nov. 2006.

\bibitem{rioul2011}
O.~Rioul, ``Information theoretic proofs of entropy power inequality,''
  \emph{IEEE Transactions on Information Theory}, vol.~57, no.~1, pp. 33--55,
  January 2011.

\bibitem{vonb}
B.~V. Bahr and C.~Esseen, ``{Inequalities for the $r$th absolute moment of a
  sum of random variables, $1 \leq r \leq 2$},'' \emph{The Annals of
  Mathematical Statistics}, vol.~36, no.~1, pp. 299--303, February 1965.

\bibitem{FAF15}
{J. Fahs and I. Abou-Faycal}, ``On the finiteness of the capacity of continuous
  channels,'' \emph{IEEE Transactions on Communications}, vol.~64, no.~1, pp.
  166--173, January 2016.

\bibitem{ShiryBook}
A.~N. Shiryaev, \emph{Probability}, 2nd~ed.\hskip 1em plus 0.5em minus
  0.4em\relax Springer-Verlag, 1996.

\bibitem{Luenb}
D.~G. Luenberger, \emph{Optimization By Vector Space Methods}.\hskip 1em plus
  0.5em minus 0.4em\relax New York: Wiley, 1969.

\bibitem{Arg12}
E.~Agrell, ``Conditions for a monotonic channel capacity,'' \emph{IEEE
  Transactions on Communications}, vol.~63, no.~3, pp. 738--748, March 2015.

\bibitem{Gal68}
R.~Gallager, \emph{Information Theory and Reliable Communication}.\hskip 1em
  plus 0.5em minus 0.4em\relax John Wiley \& Sons, Nov. 1968.

\bibitem{Hirt88}
W.~Hirt and J.~Massey, ``{Capacity of the Discrete-Time Gaussian Channel with
  Intersymbol Interference},'' \emph{Information Theory, IEEE Transactions on},
  vol.~34, no.~3, pp. 380--388, May 1988.

\bibitem{Fahs}
J.~Fahs and I.~Abou-Faycal, ``{On the detrimental effect of assuming a linear
  model for non-linear AWGN channels},'' in \emph{IEEE International Symposium
  on Information Theory}, Saint Petersburg, Russia, 2011, pp. 1693--1697.

\bibitem{Tsai2011}
Y.~Tsai, C.~Rose, R.~Song, and I.~S. Mian, ``{An Additive Exponential Noise
  Channel with a Transmission Deadline},'' in \emph{IEEE International
  Symposium on Information Theory}, Saint Petersburg, Russia, 2011, pp.
  598--602.

\bibitem{church}
{R. V. Churchill, J. W. Brown and R. F. Verhey}, \emph{Complex Variables and
  Applications}, 3rd~ed.\hskip 1em plus 0.5em minus 0.4em\relax McGraw-Hill,
  1976.

\bibitem{mil}
{J. H. Miller and J. B. Thomas}, ``{Detectors for discrete-time signals in
  non-Gaussian noise},'' \emph{Information Theory, IEEE Transactions on},
  vol.~18, no.~2, pp. 241--250, Mar. 1972.

\bibitem{shinde}
M.~P. Shinde and S.~N. Gupta, ``{Signal detection in the presence of
  atmospheric noise in tropics},'' \emph{Communications, IEEE Transactions on},
  vol.~22, pp. 1055--1063, Aug. 1974.

\bibitem{bouvet}
M.~Bouvet and S.~C. Schwartz, ``{Comparison of adaptive and robust receivers
  for signal detection in ambient underwater noise},'' \emph{Acoustics, Speech
  and Signal Processing, IEEE Transactions on}, vol.~37, pp. 621--626, May
  1989.

\bibitem{blackard}
{K. L. Blackard, T. S. Rappaport, and C. W. Bostian}, ``{Radio frequency noise
  measurments and models for indoor wireless communications at 918 MHz, 2.44
  GHz, AND 4.0 GHz},'' in \emph{IEEE International Conference on
  Communications}, vol.~1, Denver, CO, June 1991, pp. 28--32.

\bibitem{thgon}
{J. G. Gonzalez}, ``{Robust Techniques for Wireless Communications in
  non-Gaussian Environments},'' Ph.D. dissertation, University of Delaware,
  1997.

\bibitem{Shaop}
M.~Shao and C.~Nikias, ``Signal processing with fractional lower order moments:
  stable processes and their applications,'' in \emph{Proceedings of the IEEE},
  vol.~81, July 1993, pp. 986 --1010.

\bibitem{kassam}
S.~A. Kassam, \emph{Signal Detection in Non-Gaussian Noise}.\hskip 1em plus
  0.5em minus 0.4em\relax Springer-Verlag, 1988.

\bibitem{fel}
W.~Feller, \emph{An Introduction to Probability Theory and Its
  Applications}.\hskip 1em plus 0.5em minus 0.4em\relax Wiley, New York, 1966,
  vol.~2.

\bibitem{kolmo}
B.~V. Gnedenko and A.~N. Kolmogorov, \emph{Limit Distributions for Sums of
  Independent Random Variables}.\hskip 1em plus 0.5em minus 0.4em\relax Reading
  Massachusetts: Addison-Wesley Publishing Company, 1968.

\bibitem{zolo}
V.~M. Zolotarev, \emph{One-dimensional Stable Distributions}.\hskip 1em plus
  0.5em minus 0.4em\relax American Mathematical Society, 1983, vol.~65.

\bibitem{newzolo}
{V. V. Uchaikin and V. M. Zolotarev}, \emph{{CHANCE and STABILITY: Stable
  Distributions and their Applications}}.\hskip 1em plus 0.5em minus
  0.4em\relax Utrecht, Netherlands: VSP, 1999.

\bibitem{nolan:2012}
J.~P. Nolan, \emph{Stable Distributions - Models for Heavy Tailed Data}.\hskip
  1em plus 0.5em minus 0.4em\relax Boston: Birkhauser, 2012, in progress,
  Chapter 1 online at academic2.american.edu/$\sim$jpnolan.

\bibitem{ibrlin}
I.~A. Ibragimov and Y.~V. Linnik, \emph{Independent and Stationary Sequences of
  Random Variables}.\hskip 1em plus 0.5em minus 0.4em\relax Wolters-Noordhoff,
  Groningen: J.F.C. Kingman, 1971.

\bibitem{hughes}
B.~L. Hughes, ``{Alpha-stable models of multiuser interference},'' in
  \emph{IEEE International Symposium on Information Theory}, Sorrento, Italy,
  2000.

\bibitem{evans}
{K. Gulati, B. L. Evans, J. G. Andrews, and K. R. Tinsley}, ``{Statistics of
  co-channel interference in a field of Poisson and Poisson-Poisson clustered
  interferers},'' \emph{Signal Processing, IEEE Transactions on}, vol.~58,
  no.~12, pp. 6207--6222, Dec. 2010.

\bibitem{chopra}
{A. Chopra}, ``Modeling and mitigation of interference in wireless receivers
  with multiple antennas,'' Ph.D. dissertation, University of Texas at Austin,
  December 2011.

\bibitem{bc}
J.~E. Marsden and M.~J. Hoffman, \emph{Basic Complex Analysis}, 3rd~ed.\hskip
  1em plus 0.5em minus 0.4em\relax W. H. Freeman and Company, 1999.

\bibitem{bochner49}
\BIBentryALTinterwordspacing
S.~Bochner and K.~Chandrasekharan, \emph{Fourier Transforms}, ser. Annals of
  mathematics studies.\hskip 1em plus 0.5em minus 0.4em\relax Princeton
  University Press, 1949. [Online]. Available:
  \url{http://books.google.fr/books?id=zsfbTJkyp90C}
\BIBentrySTDinterwordspacing

\end{thebibliography}

\end{document}